\RequirePackage[l2tabu, orthodox]{nag}
\documentclass[11pt]{article}

\usepackage{natbib}
\usepackage{amssymb,amsmath,amsthm,latexsym,amsfonts,amscd,dsfont,enumerate}
\usepackage{graphicx, xcolor}
\usepackage{bm}
\usepackage[T1]{fontenc}
\usepackage{lmodern}
\usepackage{fixltx2e} 
\usepackage[all,error]{onlyamsmath} 
\usepackage[top=1.2in, left=0.9in, bottom=1.2in, right=0.9in]{geometry}

\definecolor{Red}{rgb}{1.00, 0.00, 0.00}

\definecolor{Blue}{rgb}{0.00, 0.00, 1.00}

\definecolor{PaleGrey}{rgb}{0.7, 0.7, 0.7}

\usepackage[normalem]{ulem}

\usepackage[]{hyperref}
 \hypersetup{
 colorlinks=true, 
 breaklinks=true, 
 urlcolor= blue, 
 linkcolor= blue, 
 bookmarksopen=true, 
 pdfauthor={Bichuch, M., Capponi, A., Sturm, S.}
 }

\newcommand{\G}{\mathcal{G}}
\newcommand{\gt}{\mathcal{G}_t}

\newcommand{\R}{\mathbb{R}}
\newcommand{\Rplus}{\mathbb R_{>0}}
\newcommand{\Px}{\mathbb{P}}
\newcommand{\Exx}{\mathbb{E}}
\newcommand{\E}{\Exx}
\newcommand{\Qxx}{\mathbb{Q}}

\newcommand{\dbarV}{\bar{\bar{{V}}}}
\newcommand{\dbarZ}{\bar{\bar{{Z}}}}
\newcommand{\uu}{\bar{\bar{w}}}
\newcommand{\hv}{\hat{\hat{w}}}
\newcommand{\vv}{\bar{v}}

\def\ind{{\mathchoice{1\mskip-4mu\mathrm l}{1\mskip-4mu\mathrm l}
{1\mskip-4.5mu\mathrm l}{1\mskip-5mu\mathrm l}}}
\newcommand{\abs}[1]{ \left \vert #1 \right \vert}
\newcommand{\define}{:=}

\newcommand{\rrp}{r_r^+}
\newcommand{\rrm}{r_r^-}
\newcommand{\rcp}{r_c^+}
\newcommand{\rcm}{r_c^-}
\newcommand{\rfp}{r_f^+}
\newcommand{\rfm}{r_f^-}

\newcommand{\XVA}{\mbox{XVA}}

\newtheorem{theorem}{Theorem}[section]
\newtheorem{definition}[theorem]{Definition}

\newtheorem{proposition}[theorem]{Proposition}
\newtheorem{remark}[theorem]{Remark}

\title{Arbitrage-Free Pricing of XVA -- Part II:\\ PDE Representation and Numerical Analysis \footnote
{The paper \cite{BCM} {subsumes the present working paper as well as} \cite{BicCapSturm}. }
}

\author{
Maxim Bichuch \thanks{Email: mbichuch@wpi.edu, Department of Mathematical Sciences, Worcester Polytechnic Institute}
\and
Agostino Capponi \thanks{Email: ac3827@columbia.edu, Industrial Engineering and Operations Research Department, Columbia University}
\and
Stephan Sturm \thanks{Email: ssturm@wpi.edu, Department of Mathematical Sciences, Worcester Polytechnic Institute}
}

\begin{document}

\maketitle

\begin{abstract}
We study the semilinear partial differential equation (PDE) associated with the non-linear BSDE characterizing buyer{'}s and seller{'}s XVA in a framework that allows for asymmetries in funding, repo and collateral rates, as well as for early contract termination due to counterparty credit risk. We show the existence of a unique classical solution to the PDE by first proving the existence and uniqueness of a viscosity solution and then its regularity. We use the uniqueness result to conduct a thorough numerical study illustrating how funding costs, repo rates, and counterparty risk contribute to determine the total valuation adjustment.
\end{abstract}

\section{Introduction}
We study the total valuation adjustment (XVA) of a European style claim under differential borrowing and lending rates. We adopt the framework introduced in the companion paper of \cite{BicCapSturm}, where the authors define buyer{'}s and seller{'}s XVA in terms of solutions to non-linear BSDEs. These characterize the portfolio process replicating long and short positions in the claim,  and take into account early termination of the contract due to counterparty credit risk, funding spreads, as well as collateral servicing costs.
When funding spreads are zero, security lending and borrowing rates coincide, and collateral rates paid by the taker equal the ones received by the provider, \cite{BicCapSturm} give explicit expressions for the XVA as well as for the corresponding replicating strategies. However, in most realistic market scenarios, rates are asymmetric and funding spreads constitute a significant driver of XVA. Despite the unavailability of closed form expressions, in this paper we show that we can exploit the connection between the BSDEs and the corresponding nonlinear PDEs to study how funding costs, collateral and counterparty risk affect the total valuation adjustment.

Our study extends previous literature along two directions. First, we develop a rigorous study of the semilinear PDEs associated with the nonlinear BSDEs introduced in \cite{BicCapSturm} and yielding the XVA prices. Previous studies in this direction include \cite{Mercurio}, who gives the PDE representations for the lower and upper bound prices of  options in an extended \cite{Piterbarg}{'}s model where borrowing and lending rates differ. However, he does not take into account counterparty risk. This is accounted for in \cite{Burgard} and \cite{BurgardCR}, who consider an extended Black-Scholes framework where two corporate bonds are introduced in order to hedge the default risk of the trader and of his counterparty. They generalize their framework in \cite{BurgardRrisk} to include collateral mitigation and evaluate the impact of different funding strategies. We also mention \cite{Crepeyimm} who study a BSDE with random terminal time of the type appearing in a funding valuation context, but do not develop connections to the resulting PDE.

We prove the existence and uniqueness of a viscosity solution to the PDE. We extend standard results of \cite{Delong} by using the locality of the viscosity solution, and also by accounting for the possibility of default risk. We then show that the PDE also admits a classical solution implying existence and uniqueness of the classical solution. Our result contrasts with earlier studies in the funding literature, e.g. \cite{BurgardCR}, where smoothness of the price of the derivative contract with respect to the underlying spot asset price is usually assumed to hold, and the PDE is derived after assuming that the market does not admit arbitrage. Using the classical solution, we compute the replicating strategies of the traded claim in a portfolio consisting of a stock, and risky bonds of the trader and of his counterparty. To the best of our knowledge, this is the first study to establish this result. \cite{Crepeya} develops a connection between the funding BSDE and the solution of a classical PDE, but in a different framework where the resulting PDE is linear and the market model complete, see Section 5.2 therein. We also mention the recent work of \cite{NiRut} who establish classical solutions to quasi-linear PDEs corresponding with BSDEs yielding the fair bilateral prices of a European contingent claim. Their work, however, does not consider the possibility that the trader or his counterparty can default.

A second noticeable contribution of our study is a comprehensive numerical analysis made possible by the previously established existence and uniqueness result. Earlier studies on funding, such as \cite{br}, \cite{Crepeya}, and \cite{BrigoPerPal} lack the computational component. We find strong sensitivity of XVA to funding costs, counterparty risk and collateralization levels. Viewing both buyer{'}s and seller{'}s XVA as functions of collateral levels defines a no-arbitrage band whose width increases with the funding spread. As the position becomes more collateralized, the value of the closeout payment increases and leads the trader to implement a riskier strategy. This yields increased buyer's and seller's XVA. Interestingly, both of these quantities may decrease in the counterparty{'}s default intensity as the default premium compensation demanded by the investor for holding counterparty bonds becomes higher than the funding costs incurred for replicating the closeout position. Moreover, increased values of counterparty's default intensity reduce both seller's and buyer's XVA by a similar amount.

Our findings serve as a useful guide to risk-management and trading desks, who can decide on the terms of the trades (collateralization levels, borrowing rates charged by treasury, etc...) based on the incurred costs as measured by XVA.

The rest of the paper is organized as follows. Section \ref{sec:model} reviews the framework for computing XVA introduced in \cite{BicCapSturm}.
Section \ref{sec:FK} analyzes the semilinear PDEs. Section \ref{sec:numanalysis} develops a comprehensive numerical analysis. Section
\ref{sec:conclusions} concludes the paper.

\section{The Setup}\label{sec:model}
We review the setup introduced in the companion paper by \cite{BicCapSturm} and refer to sections 2 and 3 of their paper for additional details. We give here a self-contained exposition needed to understand the PDE analysis presented in the next section. We fix a probability space $(\Omega,\G,{\Px})$, where $\Px$ denotes the physical measure. The market filtration is $\gt := \mathcal{F}_t \vee \mathcal{H}_t$, $t \geq 0$, where the filtrations $\mathcal{F}_t$ and $\mathcal{H}_t$ will be specified later in the section.

We take the viewpoint of an investor ``$I$'' (also referred to as trader or hedger) who wants to compute the total valuation adjustment, abbreviated as XVA, of a European claim traded with a counterparty ``$C$''. The claim is written on a stock whose time $t$ price is denoted by $S_t$. The claim matures at $T$ and has terminal payoff $\Phi(S_T)$, where $\Phi:\mathds{R}_{>0} \rightarrow \mathds{R}$ is a piecewise continuously differentiable real valued function of polynomial growth. The value of the claim as determined by a third party evaluator is given by

  \[
    \hat{V}(t,S_t) = e^{-r_D(T-t)} \mathbb{E}^{\Qxx}\bigl[ \Phi(S_T) \, \bigr\vert \, \mathcal{F}_t \bigr],
  \]
where the valuation measure $\Qxx$ associated with the publicly available discount rate $r_D$ chosen by the valuation agent is equivalent to $\Px$ and given by the Radon-Nikod\'{y}m density
  \begin{equation}\label{eq:q-girsanov}
    \frac{d\Qxx}{d\Px} \bigg|_{\mathcal{G}_{\tau}} = e^{\frac{r_D-\mu}{\sigma}W_{\tau}^{\Px} - \frac{(r_D-\mu)^2}{2\sigma^2}\tau} \Bigl(\frac{\mu_I - r_D}{h_I^{\Px}}\Bigr)^{H^I_\tau} e^{(r_D-\mu_I+h_I^{\Px})\tau}\Bigl(\frac{\mu_C - r_D}{h_C^{\Px}}\Bigr)^{H^C_\tau} e^{(r_D-\mu_C+h_C^{\Px})\tau}.
  \end{equation}

\subsection{The Portfolio securities}
This section describes the securities at disposal of the investor to construct his replicating portfolio. They include the default-free stock on which the financial claim is written, and two risky bonds underwritten by the trader and by his counterparty. Under the physical measure, the dynamics of the stock price process is given by
\[
    dS_t = \mu S_t \,dt + \sigma S_t \,dW_t^{\Px},
\]
where $\mu$ and $\sigma$ are constants. Moreover, $W^{\Px}$ is a standard Brownian motion under the probability measure $\Px$. We set $\mathbb{F} := (\mathcal{F}_t)_{t \geq 0}$ to be the $(\mathcal{G},\Px)$-augmentation of the filtration generated by $W^{\Px}$.

Let $\tau_I$ and $\tau_C$ be the default times of the trader and of his counterparty, respectively. These are independent exponentially distributed random variables with constant intensities $h_i^{\Px}$, $i \in \{I,C\}$. We use $H^i_t= \ind_{\{\tau_i\leq t\}}$, $t\geq0$, to denote the default indicator process of the $i$-th name. The default event filtration is set to $\mathcal{H}_t := \sigma(H^I_u, H^C_u \; ; u \leq t)$.

The two risky bonds underwritten by the trader $I$ and by his counterparty $C$ mature have the same maturity $T$ of the claim. For $0 \leq t \leq T$, $i \in \{I, C\}$, their price processes admit dynamics given by
  \begin{equation}\label{eq:priceproc}
    dP^i_t = (r^i+h_i^{\Px}) P^i_t \, dt - P^i_{t-} \,dH_t^i, \qquad P^i_0 = e^{-(r^i+h_i^{\Px}) T},
  \end{equation}
where $r^i$ is constant. Throughout the paper, we set $\tau := \tau_I \wedge \tau_C \wedge T$. We also recall the following relation, given in \cite{BicCapSturm}, between the physical measure $\Px$ and the valuation measure $\Qxx$:
\begin{align}
\nonumber W^{\Qxx}_t &= W^{\Px}_t + \frac{\mu-r_D}{\sigma} t \\
\nonumber \varpi_t^{i,\Qxx} &= \varpi_t^{i,\Px} + \int_0^t \bigl(1-H_u^i\bigr) (h_i^{\Px} - h_i^{\Qxx}) du, \qquad \qquad i \in \{I,C\}, \\
h_i^{\Qxx} &= h_i^{\Px} + r^i-r_D,  \qquad \qquad \qquad \qquad \qquad \qquad i \in \{I,C\}.
\label{eq:relationsmeasures}
\end{align}

\subsection{Security, Funding, and Collateral Accounts}
This section describes the various trading accounts used by the hedger to finance the strategy replicating the price process of the claim.

\subsubsection{The Security Account}
Financing and lending activities related to the stock security happen via the repo market. We denote by $r_r^{+}$ the rate received by the hedger when he lends money to the repo market. Symmetrically, we denote by $r_r^{-}$ the rate he has to pay when he borrows money from the repo market. Let $\psi_t$ be the number of shares of the security account at time $t$. The value of this account at $t$ is given by
\begin{equation}\label{eq:Brt}
    B_t^{r_r} := B_t^{r_r}\bigl(\psi \bigr) = e^{\int_0^t r_r (\psi_s) ds},
\end{equation}
where
  \begin{equation}\label{eq:rr}
    r_r(x) = r_r^{-} \ind_{\{x<0\}}+r_r^{+} \ind_{\{x>0\}}.
  \end{equation}

\subsubsection{The Funding Account}
The hedger finances the portfolio strategy to replicate the claim from his treasury. We denote by $r_f^{+}$ the rate at which the hedger lends cash to the treasury, and by $r_f^{-}$ the rate at which he borrows cash from it. Let $\xi^f_t$ denote the number of shares of the funding account at time $t$, whose value  is given by
  \begin{equation}\label{eq:Brf}
    B_t^{r_f}  := B_t^{r_f}\bigl(\xi^f) = e^{\int_0^t r_f(\xi^f_s) ds},
  \end{equation}
where
  \begin{equation}\label{eq:rrf}
    r_f  := r_f(y)= \rfm \ind_{\{{y < 0}\}}+\rfp \ind_{\{{y > 0}\}}.
  \end{equation}

\subsubsection{Collateral Account}

Each party posts collateral to mitigate potential losses incurred by the other party in case of default. Collateral is in the form of variation margins, i.e. cash transfers are made by the out-of-the money party to cover unfavourable moves in the value of the underlying transaction. On $\{\tau_I \wedge \tau_C > t\}$, the collateral process is given by
\begin{equation}\label{eq:rulecoll}
    C_t := \alpha \hat{V}(t,S_t),
\end{equation}
where the collateral level $0 \leq \alpha \leq 1$ determines the amount of covered exposure. Notice that in our setup, if the investor sold the claim to his counterparty, then he is always the collateral taker ($C_t < 0$), while if he purchased the claim from this counterparty he is always the collateral provider
($C_t > 0$). We denote by $r_c^{+}$ the rate on the collateral amount received by the hedger if he has posted collateral, and by $r_c^{-}$ the rate paid by the hedger if he has received collateral. The collateral account is defined by
  \[
    B_t^{r_c} := B_t^{r_c}(C) = e^{\int_0^t r_c(C_s) ds},
  \]
where
  \[
    r_c(x) = \rcp \ind_{\{x>0\}} + \rcm \ind_{\{x<0\}}.
  \]
We use $\psi_t^c$ to denote the number of shares of collateral account held by the hedger at time $t$. It then holds that
\begin{equation}\label{eq:collrel}
    \psi_t^{c}B_t^{r_c}  = - C_t
\end{equation}

\subsection{Close-out Value of Transaction}\label{sec:closeout}

Since both the hedger and his counterparty can default, the transaction can terminate prematurely. In the event of a default, the mark-to-market value of the transaction equals the residual value, after mitigating on-default related losses with the available collateral. Let $x^+ := \max(x,0)$, and $x^- := \max(0,-x)$, be respectively the positive and negative parts of a real number $x$. Denote by $\theta_{\tau}(\hat{V})$ the value of the transaction which must be replicated by the hedger at the earlier of the two: the hedger{'}s, or the counterparty{'}s default time (if positive the hedger owes this amount to the counterparty). This is given by
  \begin{align}
     \theta_{\tau}(\hat{V}) & \phantom{:}=  \theta_{\tau}(C,\hat{V})\label{eq:theta} \\
         \nonumber & := \hat{V}(\tau,S_{\tau}) + \ind_{\{\tau_C<\tau_I \}} L_C \bigl(\hat{V}(\tau,S_{\tau})-C_{\tau-}\bigr)^- - \ind_{\{\tau_I<\tau_C \}} L_I \bigl(\hat{V}(\tau,S_{\tau}) - C_{\tau-}\bigr)^+\nonumber
  \end{align}
Here $0 \leq L_I \leq 1$ and $0 \leq L_C \leq 1$ are constant loss rates. We refer to Section 3.3 of  \cite{BicCapSturm} for a detailed explanation.

\subsection{Total Valuation Adjustment}

We review the definition of XVA given in \cite{BicCapSturm}. Denote by $V_t({\bm\varphi})$ the wealth process of the hedger associated with the strategy ${\bm\varphi}$, where ${\bm\varphi} := \bigl(\xi_t,\xi_t^f,\xi_t^I, \xi_t^C \; t \geq 0\bigr)$. We recall from above that $\xi^{f}$ is the number of shares of the funding account. We use $\xi$ to denote the number of stock shares. Moreover, $\xi^I$ and $\xi^C$ denote the number of shares of the bonds underwritten by the trader, and by his counterparty, respectively. The wealth process $V({\bm \varphi})$ is given by the following expression
\begin{equation}\label{eq:wealth}
   V_t({\bm\varphi}) := \xi_t S_t + \xi_t^I P_t^I + \xi_t^C P_t^C + \xi_t^f B_t^{r_f} + \psi_t B_t^{r_r} - \psi_t^{c} B_t^{r_c}.
\end{equation}

\cite{BicCapSturm} show that, under the valuation measure $\Qxx$, the dynamics of the wealth process is given by

\begin{align}\label{eq:vtlast}
    \nonumber dV_t &=    \Bigl( \rfp \bigl(\xi_t^f B_t^{r_f}\bigr)^+ -\rfm \bigl(\xi_t^f B_t^{r_f}\bigr)^- +  (r_D - \rrm) \bigl(\xi_t S_t\bigr)^+ - (r_D - \rrp) \bigl(\xi_t S_t\bigr)^- + r_D \xi_t^I P_t^I + r_D \xi_t^C P_t^C\Bigr) \, dt \\
    \nonumber & \phantom{=}- \Bigl(r_c^+ \bigl(C_t\bigr)^+ - r_c^- \bigl(C_t\bigr)^-  \Bigr) \, dt  + \xi_t \sigma S_t \, dW_t^{\Qxx}  - \xi_{t-}^I P_{t-}^I   \, d\varpi_t^{I,\Qxx}  - \xi_{t-}^C P_{t-}^C d\varpi_t^{C,\Qxx} \\
  \end{align}
Setting
  \begin{equation}\label{eq:Zetas}
    \nonumber Z_t = \xi_t \sigma S_t, \qquad Z^I_t = -\xi_{t-}^I P_{t-}^I, \qquad Z_t^C = -\xi_{t-}^C P_{t-}^C,
  \end{equation}
the dynamics~\eqref{eq:vtlast} may be rewritten as
  \begin{align}\label{eq:vtlast2}
    \nonumber dV_t &= \Bigl(\rfp \bigl(V_t + Z_t^I + Z_t^C - C_t\bigr)^+ -\rfm \bigl(V_t + Z_t^I + Z_t^C - C_t\bigr)^-\\
    \nonumber & \phantom{=} +  (r_D - \rrm) \frac{1}{\sigma}\bigl(Z_t\bigr)^+ -  (r_D - \rrp) \frac{1}{\sigma}\bigl(Z_t\bigr)^- - r_D Z_t^I - r_D Z_t^C  - \Bigl(r_c^+ \bigl(C_t\bigr)^+ - r_c^- \bigl(C_t\bigr)^-  \Bigr) \,  \Bigr) dt \\
    & \phantom{=} + Z_t\, dW_t^{\Qxx} + Z_t^I \, d\varpi_t^{I,{ \Qxx}}  + Z_t^C \, d\varpi_t^{C,\Qxx}
  \end{align}
Next, we distinguish between $V_t^+$ and $-V_t^-$. We use $V_t^+$ to describe the wealth process of the hedger when he replicates the claim $\Phi(S_T)$ (hence hedging the position after selling the claim with terminal payoff $\Phi(S_T)$). On the other hand $-V_t^-$ describes the wealth process when replicating the claim $-\Phi(S_T)$ (hence hedging the position after buying the claim with terminal payoff $\Phi(S_T)$). To this purpose, define
  \begin{align}
    f^+\bigl(t,v,z,z^I,z^C; \hat{V}\bigr) &= -\Bigl(\rfp \bigl(v+ z^I + z^C- \alpha \hat{V}_t\bigr)^+ -\rfm \bigl(v + z^I + z^C - \alpha \hat{V}_t\bigr)^- \nonumber\\ & \phantom{=:}+  (r_D - \rrm) \frac{1}{\sigma}z^+  -  (r_D - \rrp) \frac{1}{\sigma}z^- - r_D z^I - r_D z^C \nonumber\\
    & \phantom{=:} + r_c^+ \bigl(\alpha \hat{V}_t\bigr)^+ - r_c^-\bigl(\alpha \hat{V}_t\bigr)^-  \Bigr)\\
    f^-\bigl(t,v,z,z^I,z^C; \hat{V}\bigr) &= - f^+\bigl(t,-v,-z,-z^I,-z^C; -\hat{V}_t\bigr)
  \end{align}
where the driver depends on the market evaluation process $(\hat{V}_t)$ (via the collateral $(C_t)$). In particular $f^\pm \, : \, \Omega \times [0,T] \times \R^4$, $(\omega, t,v,z,z^I,z^C) \mapsto f^\pm\bigl(t,v,z,z^I,z^C; \hat{V}_t(\omega)\bigr)$ are drivers of the BSDEs. Moreover, define $V^{+}$, $V^{-}$ as solutions of the BSDEs
  \begin{equation}\label{eq:BSDE-sell}
    \left\{ \begin{array}{rl} -dV_t^+ &= f^+\bigl(t,V_t^+,Z_t^+,Z_t^{I,+},Z_t^{C,+}; \hat{V}\bigr) \, dt - Z^+_t\, dW_t^{\Qxx} - Z_t^{I,+} \, d\varpi_t^{I,\Qxx}  - Z_t^{C,+} \, d\varpi_t^{C,\Qxx}\\
    V_\tau^{+} & = \theta_\tau(\hat{V}) \ind_{\{\tau<T\}} + \Phi(S_T)\ind_{\{\tau = T\}} \end{array}\right.
  \end{equation}
and
  \begin{equation}\label{eq:BSDE-buy}
    \left\{ \begin{array}{rl} -dV_t^{-} &= f^-\bigl(t,V_t^-,Z_t^-,Z_t^{I,-},Z_t^{C,-}; \hat{V}\bigr) \, dt - Z^-_t\, dW_t^{\Qxx} - Z_t^{I,-} \, d\varpi_t^{I,\Qxx}  - Z_t^{C,-} \, d\varpi_t^{C,\Qxx}\\
    V_\tau^{-} & = \theta_\tau(\hat{V}) \ind_{\{\tau<T\}} + \Phi(S_T)\ind_{\{\tau = T\}} \end{array}\right.
  \end{equation}
We have the following main theorem
\begin{theorem}[\cite{BicCapSturm}]\label{thm:arb-price}
Assume that
  \begin{equation}\label{eq:A14}
    r_r^+ \leq r_f^+ \leq r_r^-, \qquad r_f^+ \leq r_f^-, \qquad r_f^+ \vee r_D < r^I + h_I^{\Px}, \qquad r_f^+ \vee r_D < r^C + h_C^{\Px},
  \end{equation}
and
  \begin{equation}\label{eq:comp}
 \rcp \vee \rcm \le \rfm \leq \bigl(r^I + h_I^{\Px}\bigr) \wedge \bigl( r^C + h_C^{\Px}\bigr).
  \end{equation}

If $V_0^- \leq V_0^+$ then there exist prices $\pi^{sup}$ and $\pi^{inf}$, $\pi^{inf} \leq \pi^{sup}$, (called hedger's upper and lower arbitrage price) for the claim $\Phi(S_T)$ such that all prices in the closed interval $[\pi^{inf}, \pi^{sup}]$ are free of hedger's arbitrage. In particular, we have that $\pi^{sup} = V_0^+$ and $\pi^{inf} = V_0^-$.
\end{theorem}

The \textit{total valuation adjustment} $\XVA$ is defined as the amount that needs to be added to the Black-Scholes price (exclusive of funding costs) to get the actual price (inclusive of funding costs). This is asymmetric for sell- and buy-prices. \cite{BicCapSturm} define the $\XVA$s as follows.
\begin{definition}
The seller's $\XVA$ is the $\mathbb{G}$-adapted stochastic process $(\XVA_t^{sell})$ defined by
  \[
    \XVA_t^{sell} := V^+_t - \hat{V}(t,S_t)
  \]
while the buyer's $\XVA$ is defined as
  \[
    \XVA_t^{buy} := V^-_t - \hat{V}(t,S_t).
  \]
\end{definition}
$\XVA^{sell}$ corresponds to the total costs that the hedger incurs when replicating the payoff of a claim he sold, whereas $\XVA^{buy}$ corresponds to the total costs that he incurs when replicating the payoff of a claim he purchased. We note that the difference of the $\XVA$s also describes the width of the no-arbitrage interval, as
  \[
    \XVA_0^{sell} - \XVA_0^{buy} = V^+_0 - V^-_0.
  \]

\section{PDE representations of XVA} \label{sec:FK}

This section derives the PDE representations corresponding to the master BSDEs \eqref{eq:BSDE-sell} and~\eqref{eq:BSDE-buy}. For brevity, we will state everything with regards to the upper no-arbitrage price $V^{+}$, while noting that the treatment for $V^{-}$ is identical. For notational convenience, we will drop the plus superscript, and refer to $V^{+}, Z^{+}, Z^{I,+}$, $Z^{C,+}$ and $f^{+}$ simply as $V, Z, Z^{I}$, $Z^{C}$ and $f$.

\begin{remark} \label{remBSDE}
In the sequel, we write the solution of the BSDE \eqref{eq:BSDE-sell} on $\{t<\tau\}$ as $v(t,s,w^I ,w^C)$, where $ v(t, S_t, \varpi_t^{I,\Qxx},\varpi_t^{C,\Qxx})=V_t \ind_{\{\tau > t\}}$. The existence of such a measurable function $v$, i.e. the fact that $V$ is Markovian, is shown in Proposition 4.1.1 in \cite{Delong}.
\end{remark}

To simplify notation, we will not specify the argument of the latter functions. Then we can rewrite the BSDE in \eqref{eq:BSDE-sell} as
  \begin{align}
    -dV_t &=f(t, V_t, Z_t, Z_t^I, Z_t^C;\hat V_t)dt -  Z_t dW_t^{\Qxx}-Z_t^I d\varpi_t^{I,\Qxx} - Z_t^C d\varpi_t^{C,\Qxx} \label{eq:Delong-comp1}\\
    &=f(t, V_t, Z_t, Z_t^I, Z_t^C;\hat V_t)dt -  Z_t dW_t^{\Qxx}-\int_{\R} Z_t^I \tilde N^{I,\Qxx}(dt ,dr)  - \int_{\R}Z_t^C \tilde N^{C,\Qxx}(dt, dr),\nonumber\\
    V_\tau &= \theta_\tau(\hat{V}) \ind_{\{\tau<T\}} + \Phi(S_T)\ind_{\{\tau = T\}}\label{eq:Delong-comp3}
  \end{align}
where $\tilde N^{j,\Qxx},~j\in\{I,C\}$, are the compensated Poisson random measures such that
  \[
    \varpi_t^{j,\Qxx} = \int_0^t \int_{\R} \tilde N^{j,\Qxx}(ds, dr),~j\in\{I,C\},
  \]
on $[0,\tau]$ (we also refer to the proof of Theorem A.2 in \cite{BicCapSturm} for technical details). Let $\hat v(t,s)$ be the price of the claim at time $t$ conditioned on $S_t=s$, i.e. $\hat v(t,S_t) =\hat V_t$.

\begin{theorem}\label{thm:PDE}
Under the no-arbitrage conditions in \eqref{eq:A14}, $v$ is a viscosity solution of the following PDE:
  \begin{align}
    &-v_t- \sum_{j\in\{I,C\}} h_j^{\Qxx}\bigl( \theta_{j}(\hat v(t,s)) -v(t,s,w^I ,w^C)  - v_{j}\bigr)- r_D sv_s - \frac12\sigma^2 s^2v_{ss}\label{eq:PDE}\\
    &-f(t, v, \sigma sv_s(t,s,w^I ,w^C), \theta_{I}( \hat v(t,s))  - v(t,s,w^I ,w^C), \theta_{C}( \hat v(t,s)) -v(t,s,w^I ,w^C);\hat v(t,s))=0,\nonumber\\
    &v(T, s, \cdot, \cdot) = \Phi(s)\label{eq:PDE-bnd}.
  \end{align}
Here, we have used the notation $v_{i} = \frac{\partial v}{\partial w^i},~i\in\{I,C\},$ and
with slight abuse of notation $\theta_{i},~i\in\{I,C\}$, corresponds to the price of the contract at default of either party, similar to the way it was defined in \eqref{eq:theta}. Specifically,
  \begin{align*}
    \theta_{C}(\hat v) & := \hat v +   L_C ((1-\alpha) \hat{v} )^{-},  \\
    \theta_{I}(\hat v) & := \hat v - L_I  ((1-\alpha)\hat{v} )^{+}.
  \end{align*}
Additionally, $v$ is the unique viscosity solution of the PDE \eqref{eq:PDE} -- \eqref{eq:PDE-bnd} satisfying the growth condition $\lim\limits_{\abs{x}\to\infty} \abs{v(\cdot,e^x, \cdot, \cdot)} e^{-c\log^2\abs{x}}=0,~c>0.$
\end{theorem}

\begin{proof}
The following proof is an extension of the results from Theorem 4.2.2 in \cite{Delong} along two directions. The first is that, by contrast with Delong, we use the local property (as opposed to the global property) of the viscosity solution (cf. e,g, \cite[Section 7]{ShreveSoner}). The second is that we account for the occurrence of defaults, thus we need to extend the argument of  \cite{Delong} to the case when the terminal time is a stopping time, rather than constant. Concretely, we need to consider the solution on the set $\{t<\tau\}$. Let $\phi\in C^{1,2,2,2} \bigl([0,T] \times \Rplus \times \R \times \R \bigr)$ be a smooth function, such that $\phi\ge v$ and $\phi(t_0,s_0,w^I_0 ,w^C_0) = v(t_0,s_0,w^I_0 ,w^C_0)$ for some fixed $(t_0,s_0,w^I_0 ,w^C_0)$. Instead of working with $\phi$ directly, we will work with a local approximation. Let $h>0$ be small enough, such that $t_0+h< T$. As the default intensities $h_I^{\Qxx}, h_C^{\Qxx}$ are constant, it follows that $\Qxx[\tau > t_0]$ and $\Qxx[\tau > t_0 +h]$ are both strictly positive. Additionally, w.l.o.g. we may assume that $\phi$ is bounded. Specifically, let  $\{\phi_n\}_{n\ge1}$ be a subsequence such that $\phi_n\in C_b^{1,2,2,2}\bigl([0,T] \times \Rplus\times\R\times\R\bigr)$ is also bounded, satisfying $\phi_n\wedge n\vee(-n) =\phi\wedge n\vee(-n)$, and such that it converges to $\phi$ together with its derivatives uniformly on compacts. For some $n_0> \bigl\vert v(t_0,s_0,w^I_0 ,w^C_0) \bigr\vert$, we have that $\phi_{n_0} \ge v$ locally. Let $\mathcal U(t_0,s_0,w^I_0 ,w^C_0)$ be an open neighborhood of the point $(t_0,s_0,w^I_0 ,w^C_0)$ such that $\phi_{n_0}\ge v$  and (by possibly decreasing $h$) such that $(t, s_0,w^I_0 ,w^C_0)\in \mathcal U(t_0,s_0,w^I_0 ,w^C_0)$ for any $t\in[t_0, t_0+h].$ Next, choose
  \[
    n_1\ge n_0+\sup_{(t,s, w^I, w^C )\in\mathcal U(t_0,s_0,w^I_0 ,w^C_0)}\bigl\{ \bigl\vert \phi(t,s, w^I, w^C )\bigr\vert , \bigl\vert \theta_C(\hat v(t, s))\bigr\vert ,\bigl\vert \theta_I(\hat v(t, s))\bigr\vert  \bigr\}.
  \]
This allows us to make the further assumption that $(t_0, s_0,w^I_0+1 ,w^C_0)\in \mathcal U(t_0,s_0,w^I_0 ,w^C_0)$ and $(t_0, s_0,w^I_0 ,w^C_0+1)\in \mathcal U(t_0,s_0,w^I_0 ,w^C_0)$ by possibly extending the open neighborhood,  because we have that $\phi_{n_1+1}(t_0, s_0,w^I_0+1 ,w^C_0) \ge \theta_I(\hat v(t, s)),$ and $\phi_{n_1+1}(t_0, s_0,w^I_0 ,w^C_0+1) \ge \theta_C(\hat v(t, s))$. We now have that $\phi_{n_1+1}(t, S_t, \varpi_t^{I,\Qxx},\varpi_t^{C,\Qxx})\ge V_t, ~t\in[t_0, t_0+h]$, even if the default happens during this time interval. In the sequel, we will use $\phi_{n_1+1}$ in place of $\phi$, but we will abuse notation and continue referring to it as $\phi$.

For $t_0\le t\le t_0+h$ define $(\bar V,  \bar Z, \bar Z^I, \bar Z^C)$ to be a solution to the following BSDE:
  \begin{align*}
    &\bar V_{t\wedge\tau} \ind_{\{\tau > t\}} \define \phi((t_0+h)\wedge\tau,S_{(t_0+h)\wedge\tau},\varpi_{(t_0+h)\wedge\tau}^{I,{\Qxx}} ,\varpi_{(t_0+h)\wedge\tau}^{C,{\Qxx}})\ind_{\{\tau > t\}}\\
    &+\int_{t\wedge\tau}^{(t_0+h)\wedge\tau} f(r, \bar V_r, \bar Z_r, \bar Z_r^I, \bar Z_r^C;\hat V_r)dr - \int_{t\wedge\tau}^{(t_0+h)\wedge\tau} \bar Z_rdW^{^{\Qxx}}_r - \sum_{j\in\{I,C\}}  \int_{t\wedge\tau}^{(t_0+h)\wedge\tau}  \int_{\R}\bar Z_r^j\tilde N^{j,{\Qxx}}(dr, d\rho).
  \end{align*}
This BSDE has a unique solution by Theorem A.2 in \cite{BicCapSturm}, so $\bar V$ is well defined. By the comparison Theorem A.3 in \cite{BicCapSturm}, we obtain that $\bar V_{t\wedge\tau} \ind_{\{\tau > t\}}\ge V_{t\wedge\tau} \ind_{\{\tau > t\}}.$ 

For convenience, we define the operator
  \begin{align*}
    \mathcal L u(t, s, w^I ,w^C) &=  r_Dsu_s  + \int_\R \bigl( u(t, s, w^I+1 ,w^C)  - u(t, s, w^I ,w^C) - u_{I}(t, s, w^I ,w^C)\bigr)\nu^I(dz) \\
    &+ \frac{\sigma^2 s^2}2 u_{ss} +\int_\R\bigl( u(t, s, w^I ,w^C+1)  - u(t, s, w^I,w^C) - u_{C}(t, s, w^I ,w^C)\bigr)\nu^C(dz),
  \end{align*}
where the L\'evy measure $\nu^j(dx) = h_j^{\Qxx} \delta_1(dx), ~j\in\{I,C\},$ with $\delta_1$ being the Dirac measure concentrated at $1$. Additionally, let
  \begin{align*}
    \bar \Theta(t, s, w^I ,w^C) &= \phi_t(t, s, w^I ,w^C) + \mathcal L \phi(t, s, w^I ,w^C),\\
    \bar \Gamma^I(t, s, w^I ,w^C) &= \phi(t, s, w^I +1,w^C) - \phi(t, s, w^I ,w^C),\\
    \bar \Gamma^C(t, s, w^I ,w^C) &= \phi(t, s, w^I ,w^C+1) - \phi(t, s, w^I ,w^C),\\
    \dbarV_t &= \bar V_t - \phi(t, S_t, \varpi_t^{I,\Qxx} \varpi_t^{C,\Qxx}),\\
    \dbarZ_t &= \bar Z_t - \sigma S_t\phi_S(t, S_t, \varpi_{t-}^{I,\Qxx}, \varpi_{t-}^{C,{\Qxx}}),\\
    \dbarZ_t^j &= \bar Z_t^j - \bar{\Gamma}^j(t, S_t, \varpi_{t-}^{I,{\Qxx}}, \varpi_{t-}^{C,{\Qxx}}),~j\in\{I,C\}.
  \end{align*}
By It\^o's formula, we have that for $t_0\le t\le (t_0+h)$
  \begin{align*}
    &\phi(t\wedge\tau,S_{t\wedge\tau},\varpi_{t\wedge\tau}^{I,\Qxx} ,\varpi_{t\wedge\tau}^{C,\Qxx}) \ind_{\{\tau > t\}}=  \phi((t_0+h)\wedge\tau,S_{(t_0+h)\wedge\tau},\varpi_{(t_0+h)\wedge\tau}^{I,\Qxx} ,\varpi_{(t_0+h)\wedge\tau}^{C,\Qxx}) \ind_{\{\tau > t\}}\\
    &\quad- \int_{t\wedge\tau}^{(t_0+h)\wedge\tau} \bar\Theta(r, S_r,  \varpi_{r-}^{I,\Qxx}, \varpi_{r-}^{C,\Qxx}) dr -  \sigma \int_{t\wedge\tau}^{(t_0+h)\wedge\tau} S_r \phi_s(r, S_r, \varpi_{r-}^{I,\Qxx}, \varpi_{r-}^{C,\Qxx})dW_r^{\Qxx}\\
    &\quad-  \sum_{j\in\{I,C\}}\int_{t\wedge\tau}^{(t_0+h)\wedge\tau}  \int_{\R}\bar{\Gamma}^j(r, S_r, \varpi_{r-}^{I,\Qxx}, \varpi_{r-}^{C,\Qxx}) \tilde N^{j,\Qxx}(dr, d\rho).
  \end{align*}
It follows that
  \begin{align}
    \dbarV_{t\wedge\tau} \ind_{\{\tau > t\}}&= \int_{t\wedge\tau}^{(t_0+h)\wedge\tau} \Bigl( \bar\Theta(r, S_r, \varpi_{r-}^{I,\Qxx} , \varpi_{r-}^{C,\Qxx}) + f(r, \dbarV_{r} + \phi, \dbarZ_{r} + \sigma S_r\phi_s, \dbarZ_{r}^I + \bar{\Gamma}^I, \dbarZ_r^C+ \bar{\Gamma}^C;\hat V_r) \Bigr) dr\nonumber\\
    &-\int_{t\wedge\tau}^{(t_0+h)\wedge\tau}\dbarZ_rdW_r^{\Qxx} -   \sum_{j\in\{I,C\}}\int_{t\wedge\tau}^{(t_0+h)\wedge\tau}\int_{\R}\dbarZ_r^j\tilde N^{j,\Qxx}(dr, d\rho).\label{eq:viscosity1}
  \end{align}
Now, assume by contradiction that $\phi$ violates the subsolution property, i.e. there exists $\epsilon>0$ such that
  \begin{align*}
    &-\phi_t(t_0,s_0,w^I_0 ,w^C_0) - \mathcal L \phi(t_0,s_0,w^I_0 ,w^C_0) \\
    &-f(t_0,\sigma s_0\phi_s(t_0,s_0,w^I_0 ,w^C_0), \bar{\Gamma}^I(t_0,s_0,w^I_0 ,w^C_0), \bar{\Gamma}^C(t_0,s_0,w^I_0 ,w^C_0); \hat v(t_0, s_0))>\epsilon.
  \end{align*}
By continuity, and by possibly reducing $h$ and the open neighborhood of $(t_0,s_0,w^I_0 ,w^C_0)$, we may assume that it is also true inside the neighborhood $\mathcal U(t_0,s_0,w^I_0 ,w^C_0)$ of $(t_0,s_0,w^I_0 ,w^C_0)$, while still assuming that $(t, s_0,w^I_0 ,w^C_0)\in \mathcal U(t_0,s_0,w^I_0 ,w^C_0)$ for any $t\in[t_0, t_0+h].$

Next, define
  \begin{align*}
    I_{h} \define \frac{1}{h}\E^{\Qxx} \biggl[ \int_{t_0\wedge\tau}^{(t_0+h)\wedge\tau} \Psi(t,S_t,\varpi_t^{I,\Qxx} \varpi_t^{C,\Qxx})dt\biggr],
  \end{align*}
where
  \begin{align*}
    \Psi(t,s,w^I ,w^C) \define&\phi_t(t,s,w^I ,w^C)  + \mathcal L \phi(t,s,w^I ,w^C)\\
    &+ f(t, \phi,\sigma s\phi_s(t,s,w^I ,w^C), \bar{\Gamma}^I(t,s,w^I ,w^C), \bar{\Gamma}^C(t,s,w^I ,w^C);\hat v(t,s)).
  \end{align*}
Since $f$ is Lipschitz, we can show in a similar way as in \cite{Delong} that for $t_0\le t\le t_0+h$ it holds that
  \begin{align}
    &\abs{\Psi(t, s, w^I ,w^C)} \le C (1+\abs{s^2}),\nonumber\\
    &\E^{\Qxx}\Bigl[\abs{\dbarV_{t\wedge\tau}}\ind_{\{\tau > t\}}\Bigr] \le C h^\frac12,\label{eq:estim1}\\
    &\E^{\Qxx}\biggl[\int_{t_0\wedge\tau}^{(t_0+h)\wedge\tau}\abs{\dbarZ_r}^2dr  + \sum_{j\in\{I,C\}}\int_{t_0\wedge\tau}^{(t_0+h)\wedge\tau}\int_{\R}\abs{\dbarZ_r^j}^2 \nu^j(d\rho)dr\biggr] \le C h^\frac32,\label{eq:estim2}\\
    &\Qxx[\tau_1 \le (t_0+h)\wedge\tau] \le Ch,\label{eq:estim3}
  \end{align}
where we have introduced another stopping time
  \begin{align*}
    \tau_1\define \inf\bigl\{ t\ge t_0 \colon (t,S_t, \varpi_t^{I,\Qxx}, \varpi_t^{C,\Qxx}) \not\in \mathcal U(t_0,s_0,w^I_0 ,w^C_0) \bigr\}.
  \end{align*}

From $q_1\define \Qxx[ \tau >t_0]  >0$, inequality~\eqref{eq:estim3} and the fact that the default intensities $h_I^{\Qxx}, h_C^{\Qxx}$ are constant, so that $ \Qxx[ \tau \in (t_0,t_0+h)] \le Ch$, it also follows that
  \begin{align*}
    &\Qxx\bigl[\{\tau_1>(t_0+h)\wedge\tau \} \cap \{  \tau \ge t_0+h \}\bigr]\\
    &\ge 1- \Qxx[ \tau \le t_0] - \Qxx[\tau_1\le (t_0+h)\wedge\tau] - \Qxx[ \tau \in (t_0,t_0+h)]\ge q_1 - Ch,
  \end{align*}
and similarly that
  \begin{align*}
    \Qxx\bigl[\{\tau_1 \le (t_0+h)\wedge\tau\} \cup \{  \tau \in(t_0,t_0+h)\}\bigr] \le Ch.
  \end{align*}
Here and throughout this proof, we will abuse notation and use $C>0$ to denote a generic constant, which may be different in each inequality below. It follows that
  \begin{align*}
    I_{h} &= \frac1{h}\E^{\Qxx} \biggl[ \int_{t_0\wedge\tau}^{(t_0+h)\wedge\tau} \Psi(t,S_t,\varpi_t^{I,\Qxx} ,\varpi_t^{C,\Qxx})dt \ind_{\{\tau_1 >( t_0+h)\wedge\tau\} \cap \{  \tau \not\in(t_0,t_0+h)\} } \biggr] \\
    &\quad+ \frac1{h}\E^{\Qxx} \biggl[ \int_{t_0\wedge\tau}^{(t_0+h)\wedge\tau} \Psi(t,S_t,\varpi_t^{I,\Qxx} ,\varpi_t^{C,\Qxx})dt \ind_{\{\tau_1 \le (t_0+h)\wedge\tau\} \cup \{  \tau \in(t_0,t_0+h)\}} \biggr]\\
    &\le -\epsilon \Qxx\bigl[\{\tau_1>(t_0+h)\wedge\tau\} \cap \{  \tau \ge t_0+h \}\bigr] \\
    &\quad+ \frac1{h}\sqrt{\E^{\Qxx} \biggl[ \int_{t_0}^{t_0+h} 1dt \ind_{\{\tau_1 \le (t_0+h)\wedge\tau\} \cup \{  \tau \in(t_0,t_0+h)\}} \biggr] \E^{\Qxx} \biggl[   \int_{t_0\wedge\tau}^{(t_0+h)\wedge\tau} \Psi^2(t,S_t,\varpi_t^{I,\Qxx} ,\varpi_t^{C,\Qxx})dt  \biggr]}\\
    &\le -\epsilon \Qxx\bigl[\{\tau_1>(t_0+h)\wedge\tau\} \cap \{  \tau \ge t_0+\}\bigr] \\
    &\quad+ \frac{ \sqrt{h\Qxx\bigl[\{\tau_1 \le (t_0+h)\wedge\tau \} \cup \{  \tau \in(t_0,t_0+h) \}\bigr]} }{h} \sqrt{\E^{\Qxx} \biggl[ \int_{t_0\wedge\tau}^{(t_0+h)\wedge\tau} \Psi^2(t,S_t,\varpi_t^{I,\Qxx} ,\varpi_t^{C,\Qxx})dt  \biggr]}\\
    &\le -\epsilon (q_1-Ch)+ C \sqrt{{h}} \sqrt{1+\E^{\Qxx} \biggl[ \sup_{t \in [t_0\wedge\tau, (t_0+h)\wedge\tau]} S_t^4\biggr]}.
  \end{align*}
This shows that for sufficiently small $h>0$, we have that $I_h\le -\frac{\epsilon q_1}2.$

From the comparison Theorem A.3 in \cite{BicCapSturm} we get that $\dbarV_{t_0\wedge\tau} \ind_{\{\tau > t_0\}}\ge0$. Using this along with \eqref{eq:viscosity1}, we conclude that
  \begin{align*}
    \frac\epsilon2&\le \abs{\frac1{h}\dbarV_{t_0\wedge\tau}\ind_{\{\tau > t_0\}} - I_{h}}\\
    &=\frac1{h}\Bigg\vert \E^{\Qxx}\biggl[ \int_{t_0\wedge\tau}^{(t_0+h)\wedge\tau}
    f(r, \dbarV_r + \phi, \dbarZ_r + \sigma S_r \phi_s, \dbarZ_r^I + \bar{\Gamma}^I, \dbarZ_r^C+ \bar{\Gamma}^C;\hat V_r)-f(r, \phi,  \sigma S_r \phi_s,  \bar{\Gamma}^I,  \bar{\Gamma}^C;\hat V_r)dr\biggr]\Bigg\vert  \\
    &\le C \sup_{r\in[t_0,t_0+h]} \E^{\Qxx}\biggl[ \abs{\dbarV_{r\wedge\tau}}\ind_{\{\tau > r\}}\biggr] \\
    &\quad + \frac{C}{\sqrt{h}}\Biggl( \sqrt{ \E^{\Qxx}\biggl[ \int_{t_0\wedge\tau}^{(t_0+h)\wedge\tau}\abs{ \dbarZ_r}^2dr\biggr]} +  \sum_{j\in\{I,C\}}  \sqrt{ \E^{\Qxx}\biggl[\int_{t_0\wedge\tau}^{(t_0+h)\wedge\tau} \int_{\R}\abs{\dbarZ_r^j}^2 \nu^j(d\rho)dr\biggr]} \Biggr)\\
    &\le C( h^\frac12 + h^\frac14),
  \end{align*}
where the second inequality follows from the fact that the driver $f$ of our BSDE is Lipschitz along with the H\"{o}lder's inequality (setting the exponent to $2$), while the last inequality follows from the estimates given by the inequalities~\eqref{eq:estim1} and~\eqref{eq:estim2}. This yields a contradiction, as we let $h\searrow0$. Hence $\phi$ does not violate the subsolution property.

In order to incorporate the jump condition \eqref{eq:Delong-comp3} at the jump time $\tau$, note that we can rewrite the subsolution property for $\phi$ as
  \begin{align}
    &-\phi_t(t_0,s_0,w^I_0 ,w^C_0) + \sum_{j\in\{I,C\}} h_j^\Qxx\phi_j(t_0,s_0,w^I_0 ,w^C_0)- r_D s\phi_s - \frac12\sigma^2s^2 \phi_{ss} \label{eq:viscosity3}\\
    &-f_1(t_0, \phi,\sigma s_0\phi_s(t_0,s_0,w^I_0 ,w^C_0), \bar{\Gamma}^I(t_0,s_0,w^I_0 ,w^C_0), \bar{\Gamma}^C(t_0,s_0,w^I_0 ,w^C_0);\hat v(t_0,s_0))\le0,\nonumber
  \end{align}
where we set $f_1(t, v,  z,  z^I,  z^C;\hat v(t,s))\define f (t,  s, v,  z,  z^I,  z^C, \hat v(t,s))+h_I^\Qxx z^I  + h_C^\Qxx z^C.$
Under the no arbitrage conditions given by \eqref{eq:A14}, it follows that $f_1$ is an increasing function in both arguments $z^I$ and $z^C.$ By continuity of $v$ (see e.g. Lemma 4.1.1 of \cite{Delong}), it follows that $ \phi(\tau_I,S_{\tau_I},\varpi_{\tau_I-}^{I,\Qxx}+1,\varpi_{\tau_I-}^{C,\Qxx})\ind_{\{\tau_I < \tau_C \wedge T\}}\ge V_{\tau_I}\ind_{\{\tau_I < \tau_C \wedge T\}} = \theta_{I}( \hat V_{\tau_I})\ind_{\{\tau_I < \tau_C \wedge T\}},$ and similarly
$ \phi(\tau_C,S_{\tau_C},\varpi_{\tau_C-}^{I,\Qxx},\varpi_{\tau_C-}^{C,\Qxx}+1)\ind_{\{\tau_C < \tau_I \wedge T\}}\ge V_{\tau_C}\ind_{\{\tau_C < \tau_I \wedge T\}} = \theta_{C}( \hat V_{\tau_C})\ind_{\{\tau_C < \tau_I \wedge T\}}.$ Together with \eqref{eq:viscosity3} it follows that
  \begin{align*}
    &-\phi_t(t_0,s_0,w^I_0 ,w^C_0) + \sum_{j\in\{I,C\}} h_j^\Qxx\phi_j(t_0,s_0,w^I_0 ,w^C_0)- r_D s\phi_s - \frac12s^2\sigma^2 \phi_{ss}\\
    &- h_I^\Qxx  ( \theta_{I}( \hat v)  - \phi(t_0,s_0,w^I_0 ,w^C_0))- h_C^\Qxx  (\theta_{C}( \hat v) -\phi(t_0,s_0,w^I_0 ,w^C_0))\\
    &-f(t_0, \phi,\sigma s_0 \phi_s(t_0,s_0,w^I_0 ,w^C_0), \theta_{I}(\hat v)  - \phi(t_0,s_0,w^I_0 ,w^C_0), \theta_{C}(\hat v) -\phi(t_0,s_0,w^I_0 ,w^C_0); \hat v(t_0,s_0))\le0.
  \end{align*}
This shows that $\phi$ is a subsolution of the PDE~\eqref{eq:PDE} as claimed.

Finally, the uniqueness result, follows from the uniqueness of the solution to the BSDE, which in turn follows from the comparison Theorem A.3 in
\cite{BicCapSturm}.
\end{proof}

\begin{remark}\label{remark:PDE1}
Since we are only concerned with $V_t$ before any default occurs, there is no need to keep track of the martingale terms $\varpi_t^{j,\Qxx}$'s. These are only needed to realize that a default has happened. Otherwise, $\varpi^{j,\Qxx}_t = -h_j^\Qxx t.$ It then follows that $v$ reduces to a function of only two variables, so that we can simply define $\vv(t, S_t) =V_t \ind_{\{\tau>t\}}.$ In other words, $\vv(t,s) = v(t,s, -h_I^\Qxx t, -h_C^\Qxx t).$
In this case, the PDE \eqref{eq:PDE}-\eqref{eq:PDE-bnd} becomes
  \begin{align}
    &-\vv_t + (h_I^\Qxx+h_C^\Qxx)\vv(t,s)  - r_D s\vv_s - \frac12\sigma^2 s^2\vv_{ss}\label{eq:PDE1}\\
    &\quad-f(t, \vv, \sigma s\vv_s(t,s), \theta_{I}( \hat v(t,s))  - \vv(t,s) \theta_{C}( \hat v(t,s)) -\vv(t,s);\hat v(t,s))=\sum_{j\in\{I,C\}} h_j^\Qxx \theta_{j}( \hat v(t,s)),\nonumber\\
    &\vv(T, s) = \Phi(s)\label{eq:PDE-bnd1}.
  \end{align}
\end{remark}

\begin{remark}\label{remark:PDE2}
Additionally, we will also employ the standard change of variables $x=\log s, $ so that $\bar w(t,x) = \bar{v}(t,e^x)$ and $\hat w(t,x) = \hat v(t, e^x)$. Note that this change of variables allows us to get rid of boundary conditions at $S=0.$
Then, the PDE \eqref{eq:PDE1} together with the boundary condition \eqref{eq:PDE-bnd1} becomes
  \begin{align}
    &-\bar w_t - \Bigl(r_D - \frac{\sigma^2}{2}\Bigr) \bar w_x - \frac12\sigma^2\bar w_{xx} + \bigl(h_I^\Qxx+h_C^\Qxx\bigr)\bar w \label{eq:PDE2}\\
    & -f\bigl(t, \bar w, \sigma \bar w_x(t,x), \theta_{I}( \hat w(t, x))  - \bar w(t,x), \theta_{C}( \hat w(t, x)) -\bar w(t,x);\hat w(t, x)\bigr) = \sum_{j\in\{I,C\}} h_j^\Qxx \theta_{j}\bigl( \hat w(t, x) \bigr),\nonumber\\
    &\bar w(T, x) = \Phi(e^x)\label{eq:PDE-bnd2}.
  \end{align}
Hence, we can express the whole pricing problem as a Cauchy problem for a two-dimensional system of semilinear PDEs:
  \begin{align*}
    -\bar w_t + \mathcal{L}^1 \bar w &= f(t, \bar w, \sigma \bar w_x, \theta_{I}( \hat w)  - \bar w, \theta_{C}( \hat w) -\bar w;\hat w)\\
    - \hat{w}_t + \mathcal{L}^2 \hat{w} & = 0\\
    \bar w(T, x)  = \hat{w}(T, x) &= \Phi(e^x),
  \end{align*}
where the differential operators are defined by
  \begin{align*}
    \mathcal{L}^1 & := - \Bigl(r_D - \frac{\sigma^2}{2}\Bigr) \partial_x-\frac{\sigma^2}{2} \partial_{xx} + \bigl(h_I^\Qxx+h_C^\Qxx\bigr) \cdot - \sum_{j\in\{I,C\}} h_j^\Qxx \theta_{j}( \hat w)\\
    \mathcal{L}^2 & : =- \Bigl(r_D - \frac{\sigma^2}{2}\Bigr) \partial_x-\frac{\sigma^2}{2} \partial_{xx}
  \end{align*}
\end{remark}

It turns out that the PDE \eqref{eq:PDE2}-\eqref{eq:PDE-bnd2} (resp. \eqref{eq:PDE1}-\eqref{eq:PDE-bnd1}) not only has a unique viscosity solution, but this solution is also a classical solution if we assume that $\Phi$ and $\Phi'$ (where defined) have at most polynomial growth, i.e., $\abs{\Phi(s)} \le C(1+s^n) ,\abs{\Phi'(s)} \le C(1+s^n),~s\in\Rplus$ The classical argument of Theorem 20.2.1 in \cite{Cannon} assumes a bounded terminal condition, but we can employ a change of variables and divide by $(1+s^{2n})$ to utilize their framework:
\begin{align}
\uu(t,x) = \bigl(1+e^{2nx} \bigr) \bar w (t,x), \quad \hv(t,x) = \bigl(1+e^{2nx} \bigr)\hat w(t,x), \quad \bar\Phi(x) = \bigl(1+e^{2nx} \bigr)\Phi(e^x).\label{eq:var-transform}
\end{align}
In this case, our PDE  \eqref{eq:PDE1}-\eqref{eq:PDE-bnd1} becomes:
  \begin{align}
    &-\uu_t- \frac12\sigma^2 \uu_{xx} + \biggl( 2n\sigma^2\frac{e^{(2n-1)x}}{1+e^{2nx}} \uu_x -r_D \uu_x\biggr) \label{eq:PDE3}\\
    & +\biggl( h_I^\Qxx+h_C^\Qxx + 2nr_D \frac{e^{(2n-1)x}}{1+e^{2nx}} +n\sigma^2\frac{e^{(2n-2)x}}{1+e^{2nx}}\Bigl((2n-1) -4n\frac{e^{2nx}}{1+e^{2nx}}\Bigr)\biggr)\uu\nonumber\\
    & -f\biggl(t, \uu, \sigma \uu_x - 2n\sigma\frac{e^{(2n-1)x}}{1+e^{2nx}}\uu, \theta_{I}(  \hv(t,x))  - \uu, \theta_{C}( \hv(t, x)) -\uu ; \hv(t,x)\biggr) = \sum_{j\in\{I,C\}} h_j^\Qxx \theta_{j}( \hv(t, x)),\nonumber\\
    &\uu(T, x) = \bar\Phi(x)\label{eq:PDE-bnd3}.
  \end{align}

We can then prove the following
\begin{proposition}\label{prop:sol-exists}
Assume that $\Phi$ is piecewise continuously differentiable and $\Phi$ as well as $\Phi'$ (where defined) have at most polynomial growth, i.e.,  $\abs{\Phi(s)} \le C(1+s^n),\abs{\Phi'(s)} \le C(1+s^n),~s\in\Rplus.$ Then the PDE \eqref{eq:PDE2} with terminal condition \eqref{eq:PDE-bnd2} has a classical solution.
\end{proposition}

\begin{proof}
Assume first that $\Phi$ is continuosly differentiable. Using the transformation \eqref{eq:var-transform} given above, it is then suffices to prove that \eqref{eq:PDE3}-\eqref{eq:PDE-bnd3} has a classical solution. The above transformation guarantees that both $\bar \Phi$ and $\bar\Phi'$ are bounded. Then the existence of a smooth (and bounded) solution to \eqref{eq:PDE3}-\eqref{eq:PDE-bnd3} follows from Theorem 20.2.1 in \cite{Cannon}. In the case that $\bar\Phi$ is only piecewise smooth, the original proof can be modified following a similar procedure to \cite{JK}. Hence, using the change of variables \eqref{eq:var-transform}, we conclude that there exists a classical solution to the PDE \eqref{eq:PDE2} with terminal condition \eqref{eq:PDE-bnd2}.
\end{proof}

Combining Proposition~\ref{prop:sol-exists} with the uniqueness result from Theorem \ref{thm:PDE} allows us to conclude that there exists a unique classical solution to the PDE with polynomial growth. Moreover, the hedging strategies can be easily found by an application of Theorem \ref{thm:strat}. Using the relations in~\eqref{eq:Zetas} we find that on the set $\{t<\tau\}$
  \begin{align}
    \nonumber \xi_t &= \frac{Z_t}{\sigma S_t} = \vv_S(t, S_t),\\
    \xi_{t}^j&=-\frac{Z_t^j}{P_{t}^j} = \frac{\vv(t, S_t) - \theta_j(\hat v(t,S_t))}{e^{-(r_D + h_j^{\Qxx}) (T-t)}},~j\in\{I,C\}. \label{eq:strat}
  \end{align}
where we have used that the bond price $P_{t}^j = e^{(r_D + h_j^{\Qxx}) t}$ on $\tau_j > t$, by virtue of \eqref{eq:priceproc} and the relations between the measures $\Px$ and $\Qxx$ given in \eqref{eq:relationsmeasures}.

\section{Numerical Analysis} \label{sec:numanalysis}

We perform a comparative statics analysis to analyze the dependence of XVA and portfolio replicating strategies on funding rates, default intensities, and collateral levels. We consider the relative XVA, i.e. express the adjustment as a percentage of the price $\hat{V}_t$ of the claim, given by $\frac{V_t - \hat{V}_t}{\hat{V}_t}$, where $V_t =V_t^{\pm}$ depending on whether we are considering buyer{'}s or seller{'}s XVA. The claim is chosen to be a European-style call option on the stock security, i.e. $\Phi(x) = (x-K)^+$. We consider one at-the-money option, with $S_0=K=1$ maturing at $T=1$. In order to focus on the impact of funding costs (which in practice is the most relevant) and separate it from additional contributions to the XVA coming from asymmetries in collateral and repo rates, we set $r_r^+ = r_r^- = 0.05$, and $r_c^+ = r_c^- = 0.01$.
As the derivatives contract does not only specify the price of the option but also the levels of collateralization of the deal, the no-arbitrage region appears as a (two-dimensional) band in $\XVA$ and $\alpha$ rather than as a (one-dimensional) interval in $\XVA$ only.

We also use the following other benchmark parameters: $\sigma = 0.2$, $r_f^+ = 0.05$, $r_f^- = 0.08$, $r_D = 0.01$, $r^I = 0.03$, $r^C = 0.04$,  $h_I^{\Qxx} = 0.2$, $h_C^{\Qxx} = 0.15$, $L_I = L_C = 0.5$, and $\alpha = 0.9$. We compute the numerical solution of the PDE using a finite difference Crank-Nicholson scheme.

The main finding of our analysis are discussed in the sequel:

\paragraph{Higher funding rates increase the width of the no-arbitrage band.}

Figure \ref{fig:alpharfm} displays the no-arbitrage band whose width is increasing in the funding rate $r_f^-$. As $\alpha$ gets higher, the band noticeably shrinks reaching its minimum around $\alpha = 80\%$ before widening again. Notice that buyer{'}s and seller{'}s XVA do not have a symmetric behavior. This can be better understood by analyzing the dependence of the band on the collateral level $\alpha$ in Figure \ref{fig:alpharfm}. If $\alpha$ is not too high ($\alpha < 0.5$), the widening of the no-arbitrage band with respect to the funding rate $r_f^-$ is due to decreasing buyer{'}s XVA. On the other hand, if $\alpha$ is high the buyer{'}s XVA is insensitive to changes in $r_f^-$ whereas the seller{'}s XVA increases with $r_f^-$, contributing to widen the no-arbitrage band. This is further supported by the numerical values reported in Table \ref{tab:Tablealrfm}. When $\alpha < 0.5$, the position in the funding account for the seller{'}s XVA is long and the same regardless of the funding rate $r_f^-$. On the other hand, the size of the long position for the buyer{'}s XVA increases in $r_f^-$. In presence of full collateralization, i.e. $\alpha=1$, the situation reverses. The position in the funding account for the buyer{'}s XVA is short and stays constant with respect to $r_f^-$. Vice versa, for the seller{'}s XVA the size of the short position increases in size with respect to $r_f^-$.

If $\alpha$ is high, the trader will have to post more collateral and consequently reduce the cash resources for his replicating strategy. He will then have to borrow more from the funding desk, resulting in higher funding costs. This drives up both the seller{'}s XVA and the number of shares of stocks and bonds needed for the replication strategy.

  \begin{figure}[ht!]
    \centering
      \includegraphics[width=6.6cm]{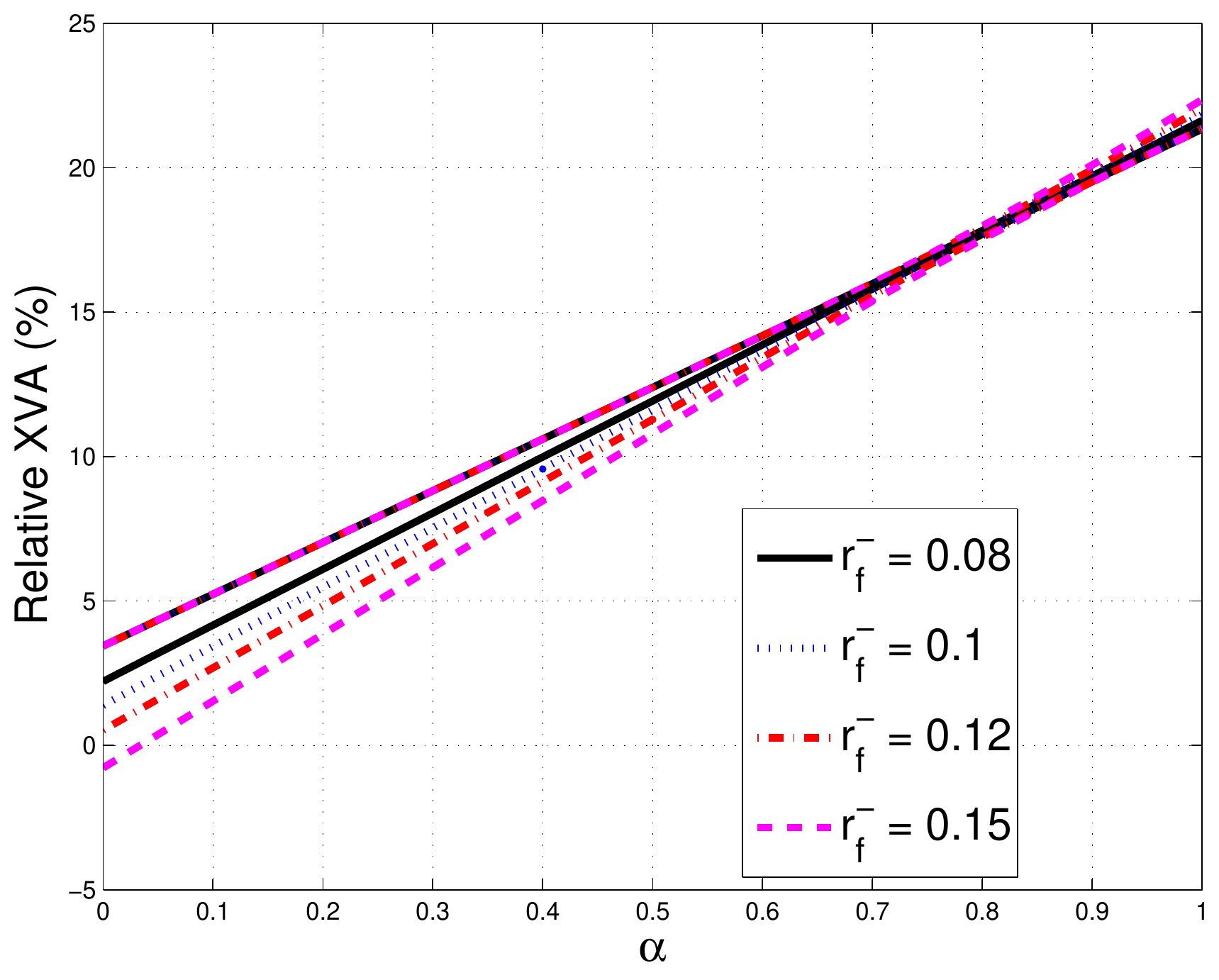}
      \includegraphics[width=6.6cm]{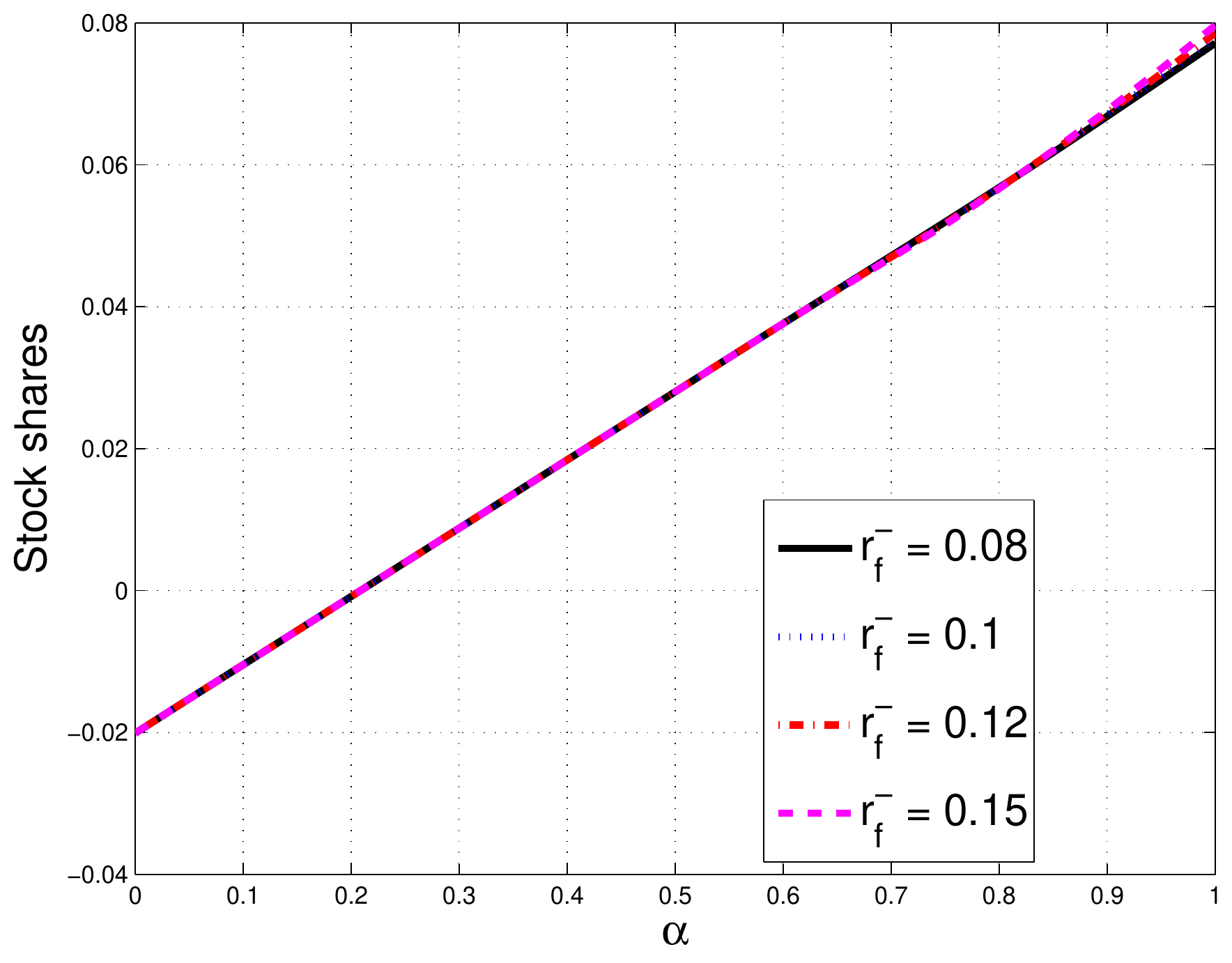}
      \includegraphics[width=6.6cm]{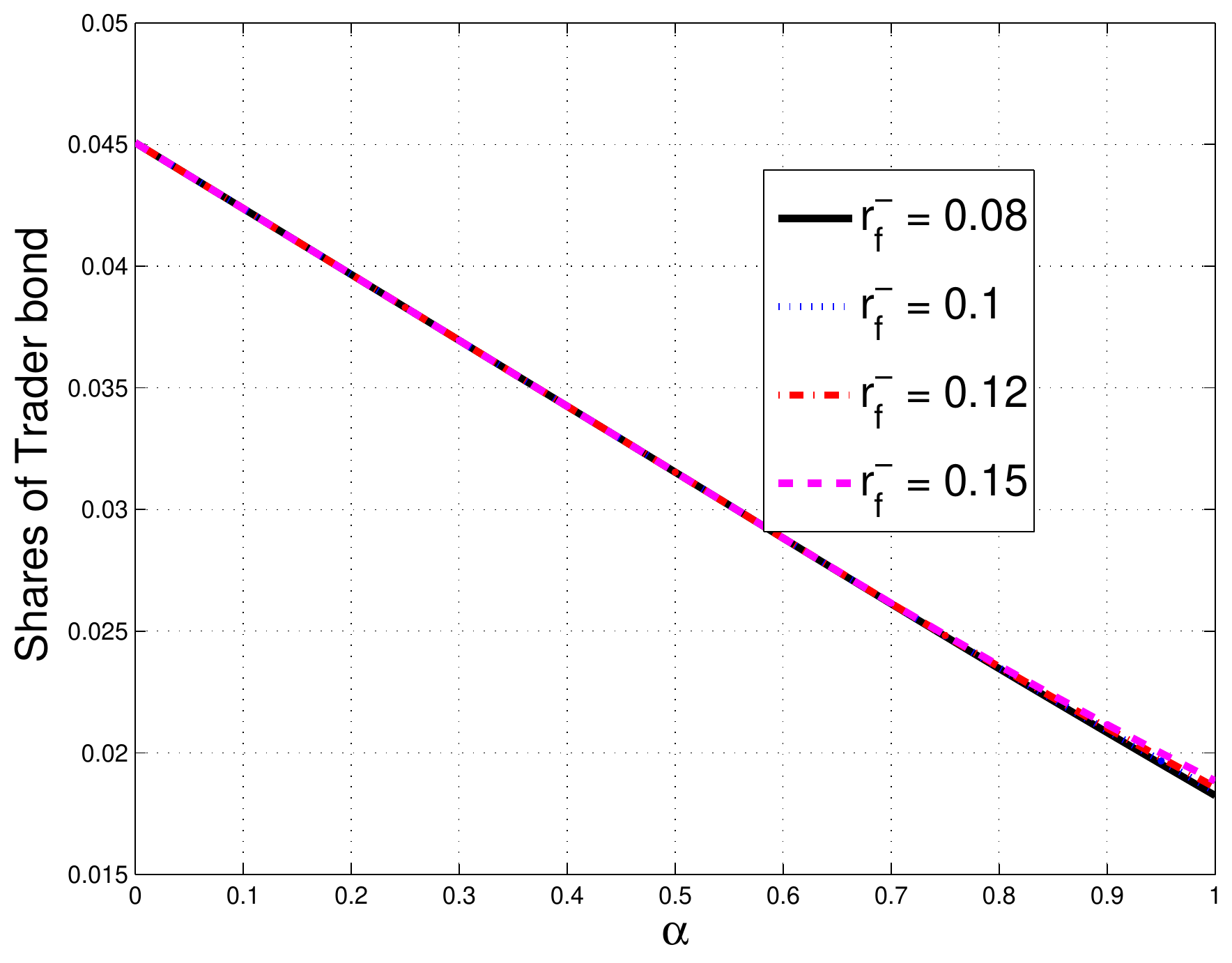}
      \includegraphics[width=6.6cm]{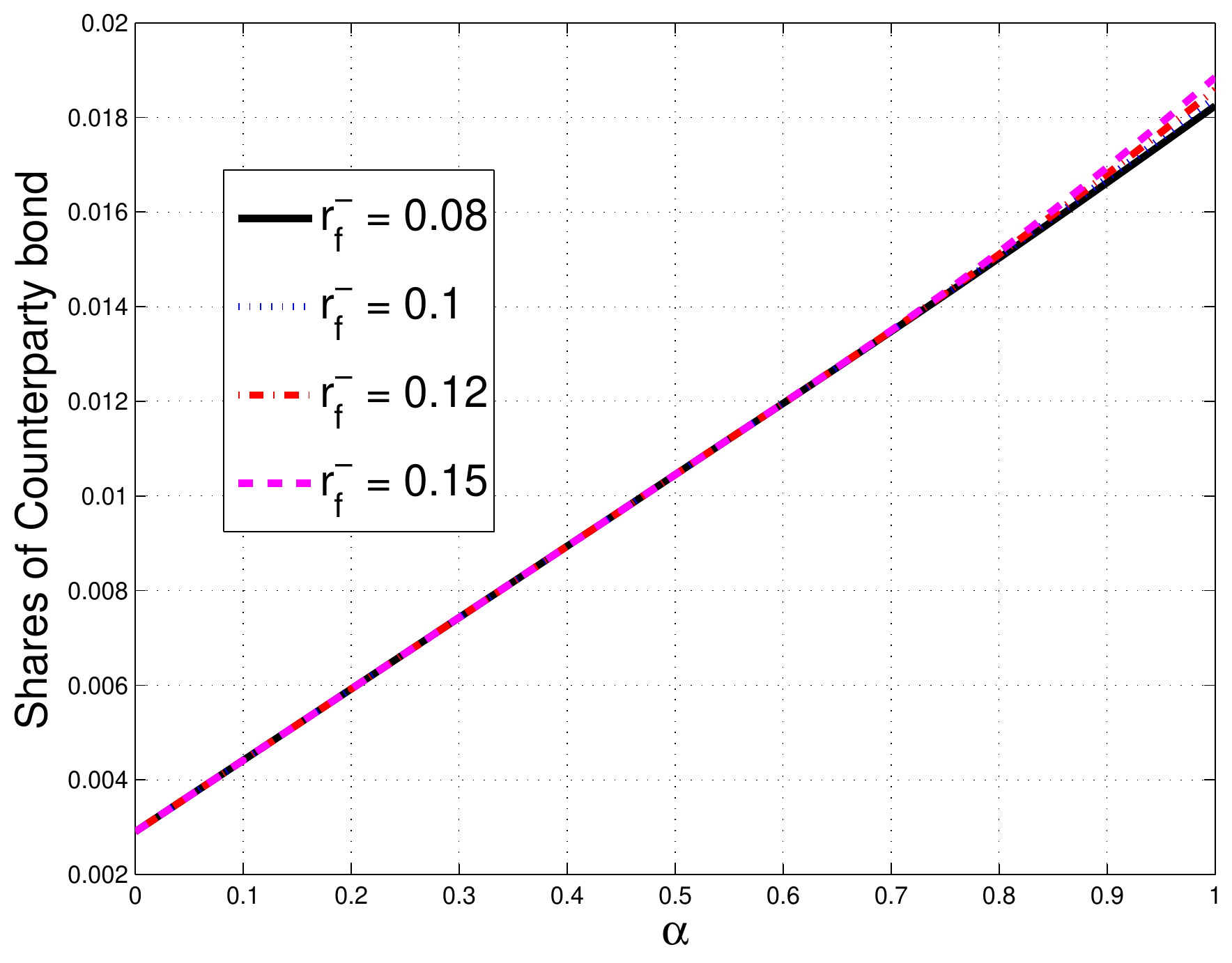}
    \caption{Top left: Buyer{'}s and seller{'}s XVA as a function of $\alpha$ for different $r_f^-$. The seller's lies above the buyer{'}s XVA and the same line style is used for both. Top right: Number of stock shares in the replication strategy. Bottom left: Number of trader bond shares in the replication strategy. Bottom right: Number of counterparty bond shares in the replication strategy. We plot the strategies for the portfolio replicating the seller{'}s XVA.}
  \label{fig:alpharfm}
  \end{figure}

  \begin{table}[hpt]
    \centering
      \begin{tabular}{|c|c|c|c|}
	\hline
	$\alpha$ & $r_f^-$ & {\text Seller{'}s XVA: funding account} (\$) & {\text Buyer{'}s XVA: funding account} (\$)\\
	\hline
	\hline
	0 & 0.08 & 0.0039 & 0.0403 \\
	\hline
	0 & 0.2 & 0.0039 &  0.0447 \\
	\hline
	0.25 & 0.08 & 0.0249 & 0.0257 \\
	\hline
	0.25 & 0.2 & 0.0249 & 0.0287 \\
	\hline
	0.75 & 0.08 & -0.0037 & -0.0036 \\
	\hline
	0.75 & 0.2 & -0.0038 & -0.0032 \\
	\hline
	1 & 0.08 & -0.0182 & -0.018 \\
	\hline
	1 & 0.2 & -0.0193 & -0.018 \\
	\hline
      \end{tabular}
    \caption{The columns give the dollar position in the funding account corresponding to the replicating strategies of seller{'}s XVA and buyer{'}s XVA.}
  \label{tab:Tablealrfm}
  \end{table}

\paragraph{Higher collateralization increases portfolio holdings.}

  \begin{figure}[ht!]
    \centering
      \includegraphics[width=6.6cm]{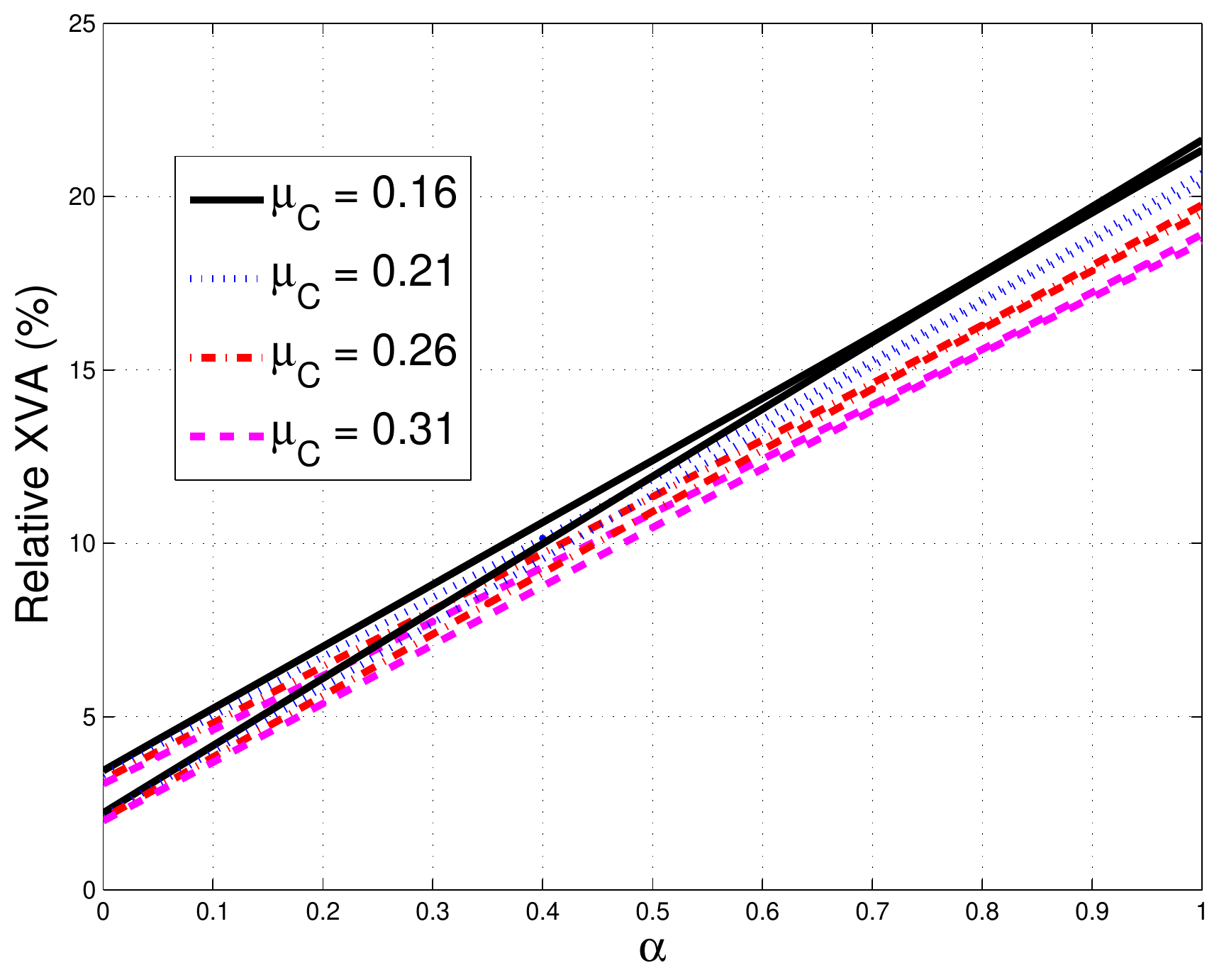}
      \includegraphics[width=6.6cm]{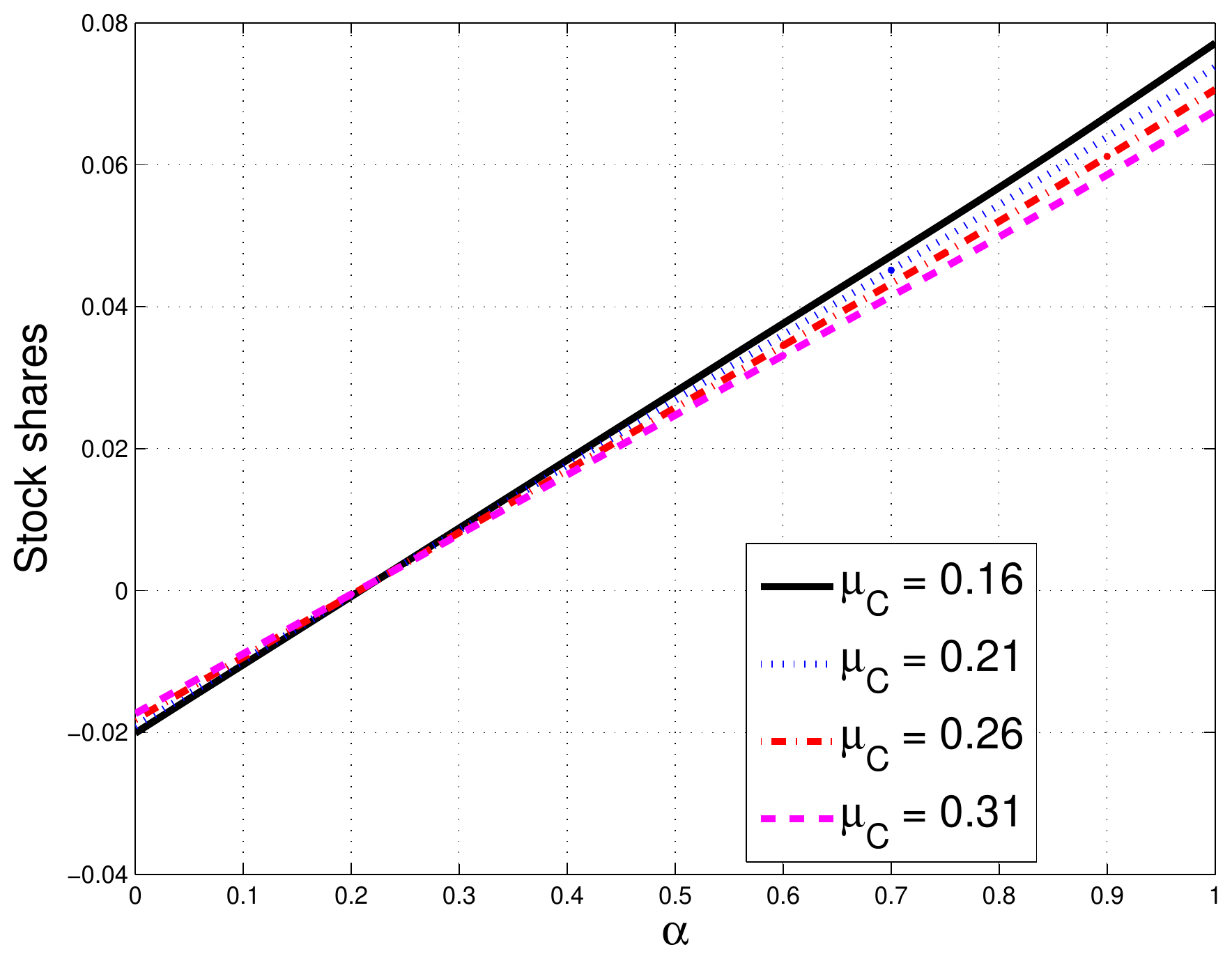}
      \includegraphics[width=6.6cm]{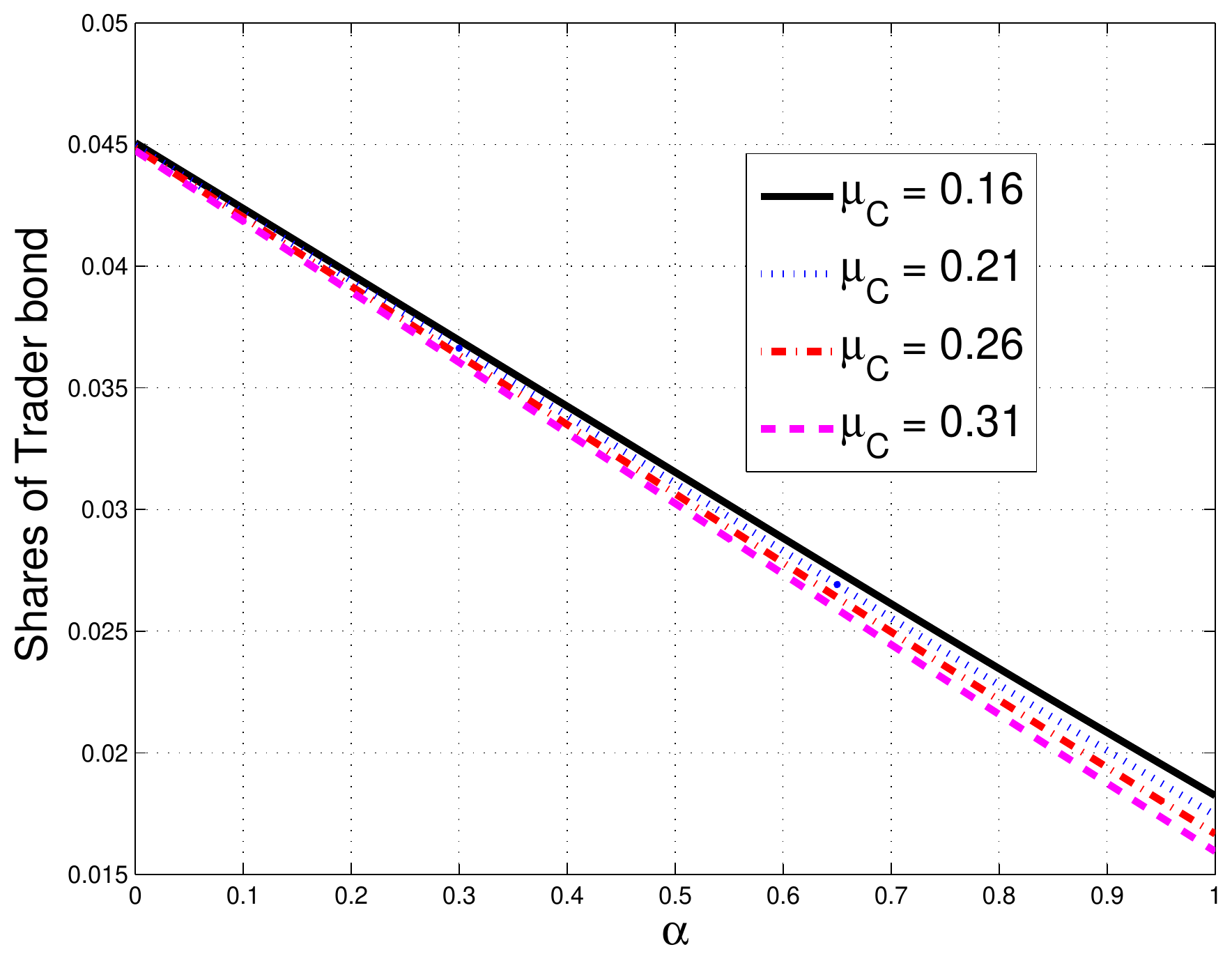}
      \includegraphics[width=6.6cm]{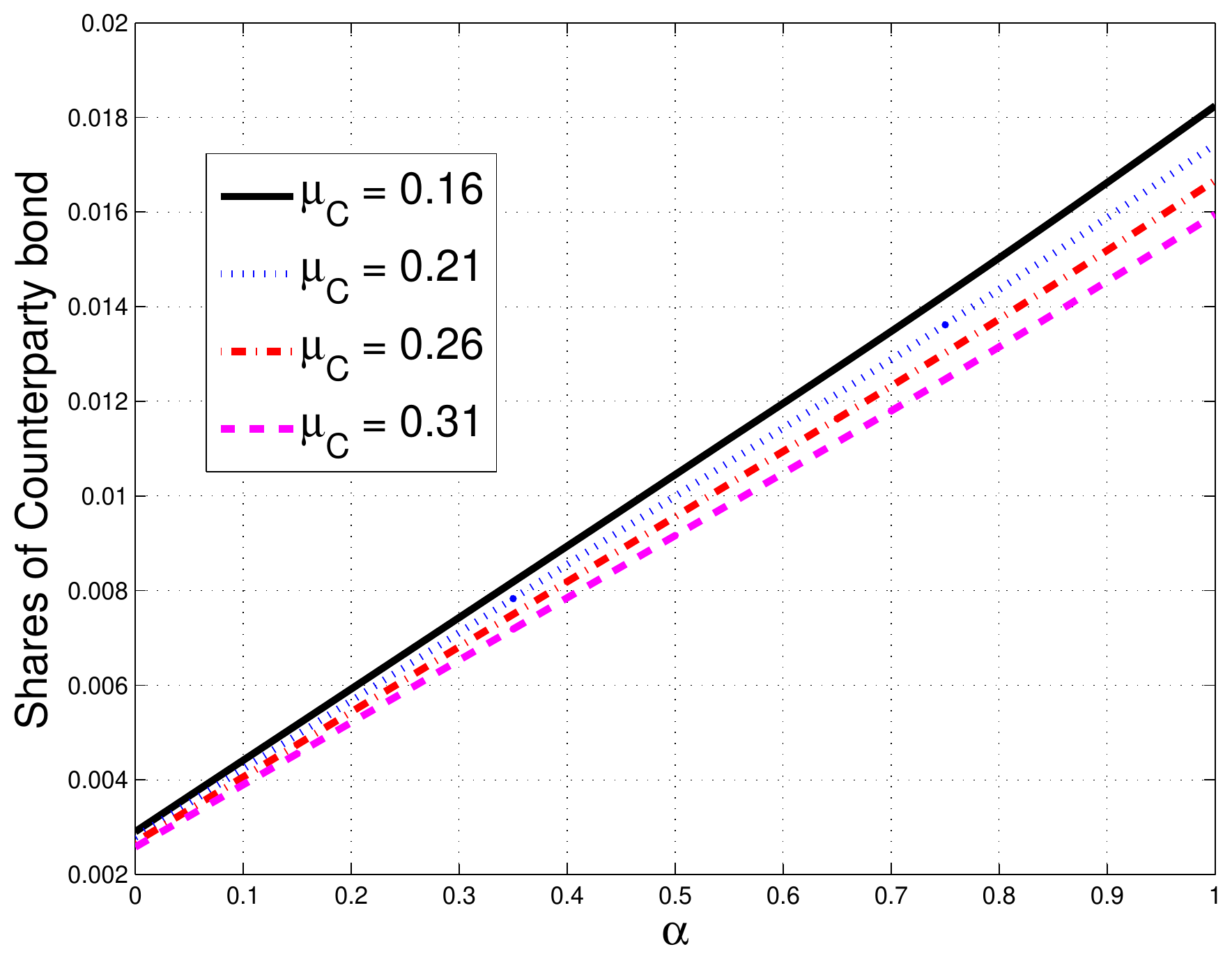}
    \caption{Top left: Buyer{'}s and seller{'}s XVA  as a function of $\alpha$ for different $h_C^{\Qxx}$. Top right: Number of stock shares in the replication strategy. Bottom left: Number of trader bond shares in the replication strategy. Bottom right: Number of counterparty bond shares in the replication strategy.}
  \label{fig:alphahc}
  \end{figure}

As the collateral level $\alpha$ increases, the seller{'}s XVA increases. This happens because the value of the closeout position becomes higher as it can be directly seen from Eq.~\eqref{eq:theta} (notice that $\hat{V} > 0$ because we are considering a short call option position). The trader would then need to construct a portfolio replicating a larger position, hence he must take more risk. He achieves this by increasing the number of shares of stock and bond underwritten by the counterparty. Moreover, higher collateralization levels reduce the size of the downward negative jump to the closeout value occurring when the trader defaults. Consequently, the trader needs to purchase a smaller amount of his bonds to replicate this position as $\alpha$ increases. This behavior is confirmed from the plot in Figure \ref{fig:alphahc}.

\paragraph{The width of the no-arbitrage band is insensitive to counterparty{'}s default intensity. }
Figure~\ref{fig:alphahc} shows that both seller{'}s and buyers{'}s XVA decrease, if the counterparty{'}s default intensity $h_C^{\Qxx}$ increases. When $\alpha$ is low, the two quantities drop by nearly the same amount and the width of the no-arbitrage band is unaffected. As $\alpha$ gets larger, the seller{'}s XVA decreases faster relative to the buyer{'}s XVA and the two quantities almost coincide when $\alpha=1$.

  \begin{table}[hpt]
    \centering
      \begin{tabular}{|c|c|c|}
	\hline
	$r_f^-$ & {\text Seller{'}s XVA: funding} (\$) & {\text Buyer{'}s XVA: funding account} (\$)\\
	\hline
	\hline
	0.08 & -0.0124 & -0.0123 \\
	\hline
	0.1 &  -0.0125 & -0.0122 \\
	\hline
	0.15 & -0.0127 & -0.0122 \\
	\hline
	0.2 & -0.013 & -0.0122  \\
	\hline
      \end{tabular}
    \caption{The columns give the dollar position in the funding account corresponding to the replicating strategies of seller{'}s XVA and buyer{'}s XVA. We set $h^{\Qxx}_C = 0.15$.}
  \label{tab:Tablerf}
  \end{table}

Consistently with Figure \ref{fig:alphahc}, Figure \ref{fig:hcalpha} shows that the seller{'}s XVA decreases when the default intensity of the counterparty $h^{\Qxx}_C$ increases. This can be understood as follows. Using the relations between the default intensities under the probability measures $\Px$ and $\Qxx$ given in Eq.~\eqref{eq:relationsmeasures}, we can see that increasing $h_C^{\Qxx}$ is equivalent to increasing the bond rate $r^C$ while keeping the default probability under $\Px$ constant. Hence, the trader would earn higher premium from his long position in counterparty bonds (see also bottom panels of Figure \ref{fig:hcalpha}). Such a gain dominates over the funding costs incurred when replicating a larger closeout position (Eq.~\eqref{eq:theta} indicates that the closeout payment increases to the risk-free payoff $\hat{V}$ as $h^{\Qxx}_C$ increases). Altogether, this means that the funding costs of the investor would be reduced as $h_C^{\Qxx}$ increases.

Hence, the dependence of XVA on counterparty{'}s default intensity contrasts with its sensitivity to collateral levels numerically illustrated in Figure \ref{fig:alpharfm}. This is because higher default risk of the counterparty also means higher return on the bond underwritten by the counterparty, whereas higher $\alpha$ only means that a larger value of the closeout payment needs to be replicated.

  \begin{figure}[ht!]
    \centering
      \includegraphics[width=6.6cm]{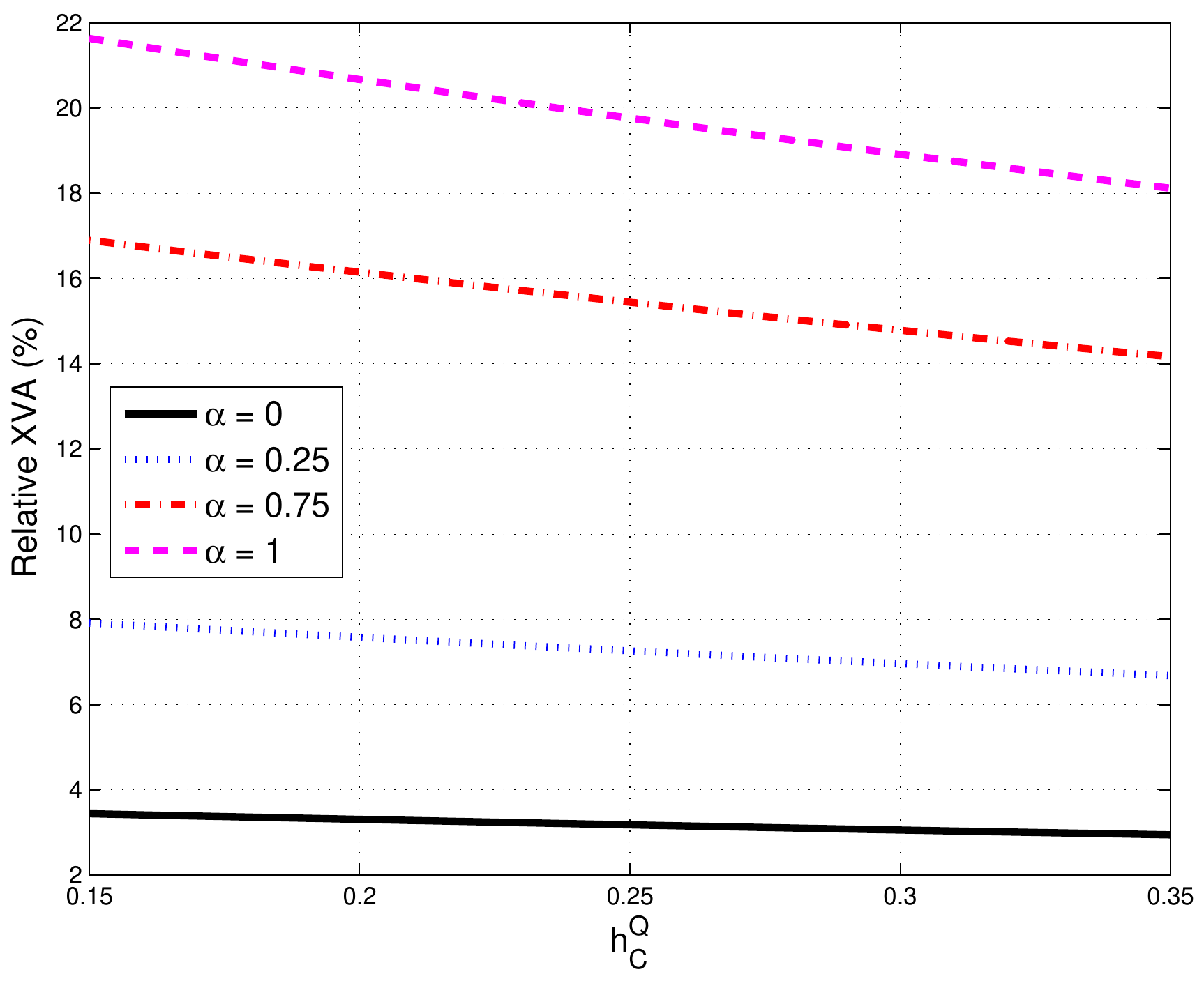}
      \includegraphics[width=6.6cm]{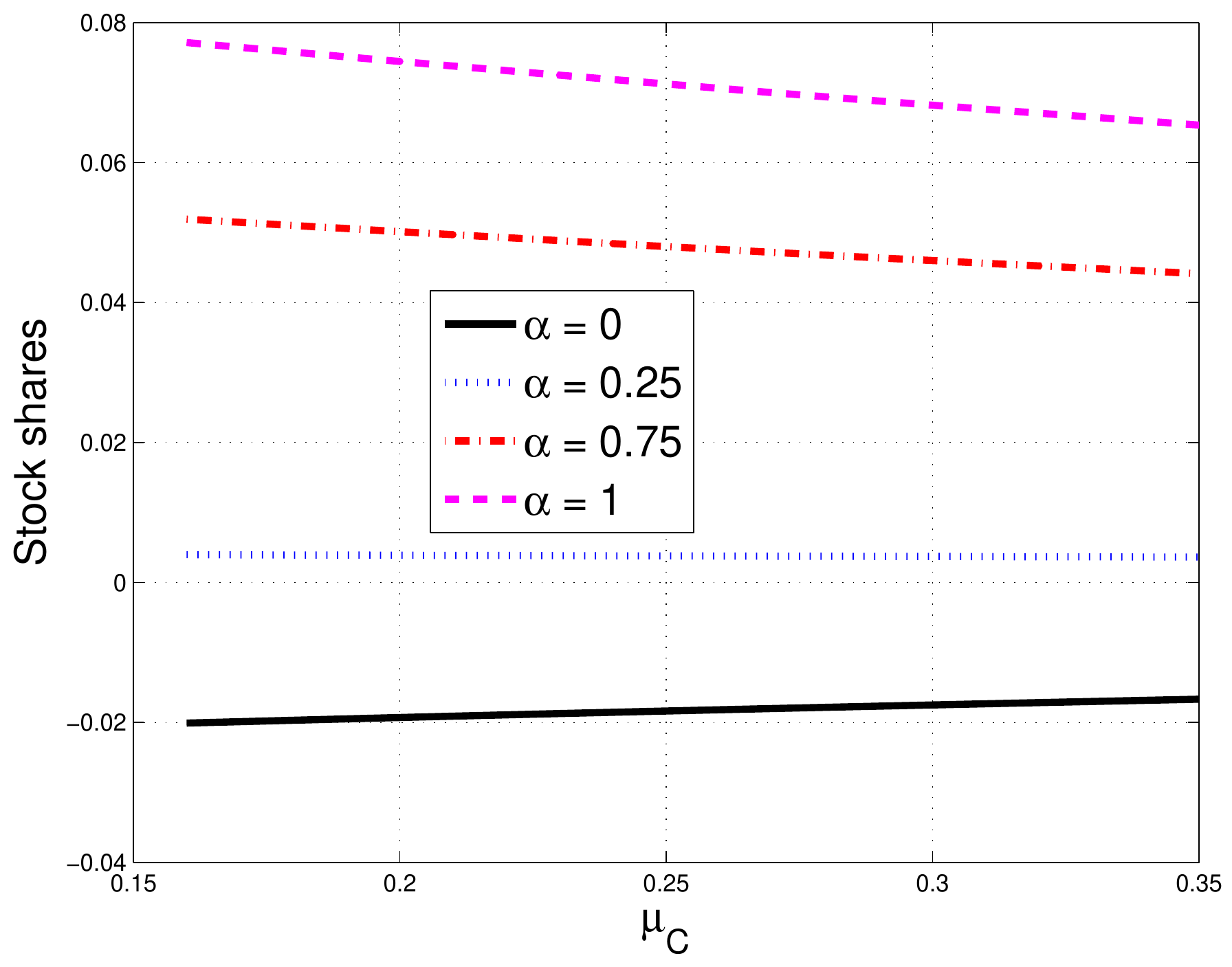}
      \includegraphics[width=6.6cm]{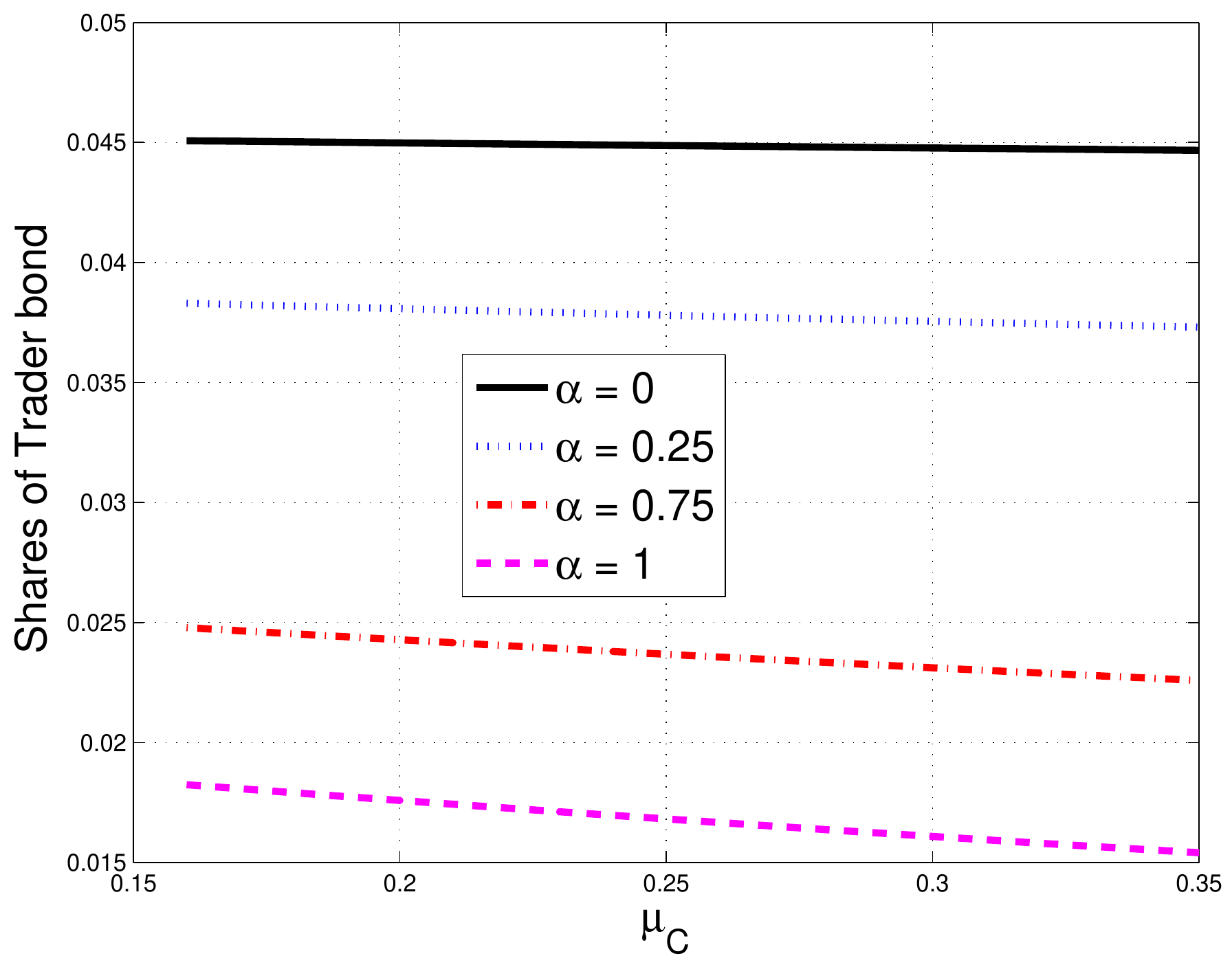}
      \includegraphics[width=6.6cm]{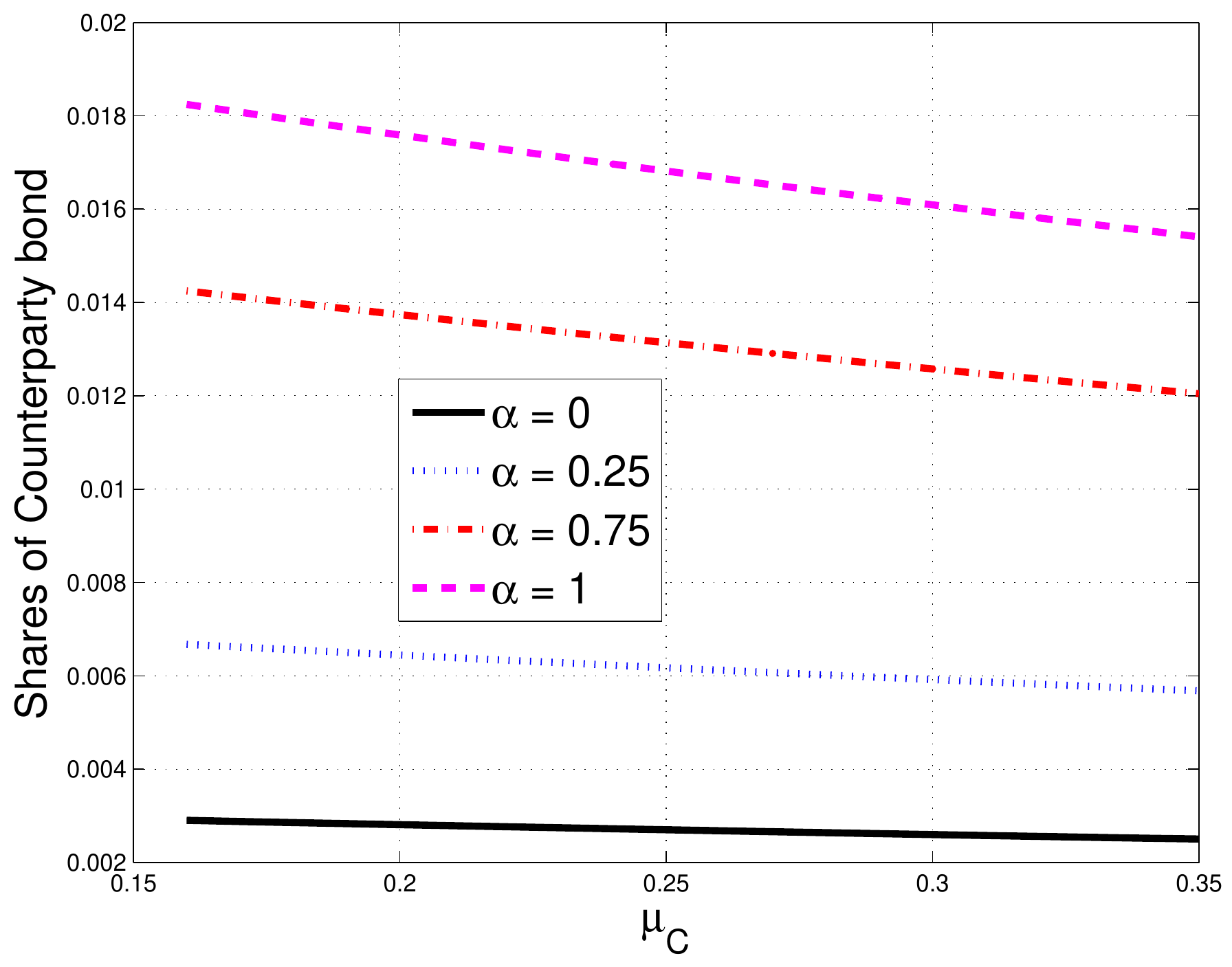}
    \caption{Top left: Seller{'}s XVA as a function of $h_C^{\Qxx}$ for different $\alpha$. Top right: Number of stock shares in the replication strategy. Bottom left: Number of trader bond shares in the replication strategy. Bottom right: Number of counterparty bond shares in the replication strategy. The replicating portfolio refers to the seller{'}s XVA.}
  \label{fig:hcalpha}
  \end{figure}

\section{Conclusions} \label{sec:conclusions}
We have developed a rigorous analysis of the semilinear PDE associated with the BSDE characterizing the price process of a portfolio replicating a European option, when funding, collateral and closeout costs are taken into account. We have shown the existence and uniqueness of a classical solution to the PDE under mild assumptions on the coefficients. Using this result, we have conducted a thorough numerical study analyzing the sensitivity of XVA and of the claim{'}s replication strategy to collateral levels, default risk and rates asymmetries. Our findings support the introduction of centralized XVA desks to manage and hedge all costs related to over-the-counter transactions. It shows that funding costs originating from the different trading components cannot be easily separated and hence attributed to different business units (CVA, DVA and FVA desks) because they are highly interdependent.

\appendix
\section{PDEs and Replication Strategies}
Recall the measurable function $\vv$ from Remark \ref{remark:PDE1} that was defined by $\vv(t, S_t) =V_t \ind_{\{\tau>t\}}.$ The following theorem shows how the function $\vv$ can be used to compute the hedging strategies. Noting that the law of $S_t$ is absolutely continuous, the proof of the theorem becomes analogous to the proof of Theorem 4.1.4 in \cite{Delong}. Hence, we omit it here and only give the statement of the theorem.

\newpage

\begin{theorem}\label{thm:strat}
Consider the data $(f,\theta_\tau(\hat V(\tau, S_\tau)),\hat V(T, S_T))$ for the BSDE given by \eqref{eq:BSDE-sell}. Additionally, let the function $\vv(t, S_t) =V_t \ind_{\{\tau>t\}}$ be defined as in
Remark \ref{remark:PDE1}. Then, on the set $\{t<\tau\}$ we have that
  \begin{align*}
    Z_t &= \sigma S_t \vv_S(t,S_t),\\
    Z_t^j &= \theta_j(\hat V(t, S_t)) - \vv(t,S_t).
  \end{align*}
\end{theorem}

\end{document}